\PassOptionsToPackage{unicode}{hyperref}
\PassOptionsToPackage{hyphens}{url}
\PassOptionsToPackage{dvipsnames,svgnames,x11names}{xcolor}
\documentclass[
  12pt]{article}

\usepackage{amsmath,amssymb,amsthm,amsfonts,bm,mathtools}
\usepackage{booktabs}
\usepackage{enumitem}
\usepackage{algorithm}
\usepackage{algpseudocode}
\usepackage{iftex}
\ifPDFTeX
  \usepackage[T1]{fontenc}
  \usepackage[utf8]{inputenc}
  \usepackage{textcomp} 
\else 
  \usepackage{unicode-math}
  \defaultfontfeatures{Scale=MatchLowercase}
  \defaultfontfeatures[\rmfamily]{Ligatures=TeX,Scale=1}
\fi
\usepackage{lmodern}
\ifPDFTeX\else  
\fi
\IfFileExists{upquote.sty}{\usepackage{upquote}}{}
\IfFileExists{microtype.sty}{
  \usepackage[]{microtype}
  \UseMicrotypeSet[protrusion]{basicmath} 
}{}
\makeatletter
\@ifundefined{KOMAClassName}{
  \IfFileExists{parskip.sty}{%
    \usepackage{parskip}
  }{
    \setlength{\parindent}{0pt}
    \setlength{\parskip}{6pt plus 2pt minus 1pt}}
}{
  \KOMAoptions{parskip=half}}
\makeatother
\usepackage{xcolor}
\makeatletter
\ifx\paragraph\undefined\else
  \let\oldparagraph\paragraph
  \renewcommand{\paragraph}{
    \@ifstar
      \xxxParagraphStar
      \xxxParagraphNoStar
  }
  \newcommand{\xxxParagraphStar}[1]{\oldparagraph*{#1}\mbox{}}
  \newcommand{\xxxParagraphNoStar}[1]{\oldparagraph{#1}\mbox{}}
\fi
\ifx\subparagraph\undefined\else
  \let\oldsubparagraph\subparagraph
  \renewcommand{\subparagraph}{
    \@ifstar
      \xxxSubParagraphStar
      \xxxSubParagraphNoStar
  }
  \newcommand{\xxxSubParagraphStar}[1]{\oldsubparagraph*{#1}\mbox{}}
  \newcommand{\xxxSubParagraphNoStar}[1]{\oldsubparagraph{#1}\mbox{}}
\fi
\makeatother

\usepackage{longtable,booktabs,array,tabularx}
\usepackage{setspace}
\newcolumntype{L}[1]{>{\raggedright\arraybackslash}p{#1}}
\newcommand{\TableNote}[1]{%
  \par\vspace{0.35em}\noindent
  \begin{minipage}{0.96\linewidth}
    \footnotesize\setstretch{1.05}\raggedright
    \textit{Notes:} #1
  \end{minipage}%
  \par\vspace{0.6\baselineskip}%
}

\usepackage{calc} 
\usepackage{etoolbox}
\makeatletter
\patchcmd\longtable{\par}{\if@noskipsec\mbox{}\fi\par}{}{}
\makeatother
\IfFileExists{footnotehyper.sty}{\usepackage{footnotehyper}}{\usepackage{footnote}}
\makesavenoteenv{longtable}
\usepackage{graphicx}
\makeatletter
\def\maxwidth{\ifdim\Gin@nat@width>\linewidth\linewidth\else\Gin@nat@width\fi}
\def\maxheight{\ifdim\Gin@nat@height>\textheight\textheight\else\Gin@nat@height\fi}
\makeatother
\setkeys{Gin}{width=\maxwidth,height=\maxheight,keepaspectratio}
\makeatletter
\def\fps@figure{htbp}
\makeatother

\makeatletter
\@ifpackageloaded{caption}{}{\usepackage{caption}}
\AtBeginDocument{%
\ifdefined\contentsname
  \renewcommand*\contentsname{Table of contents}
\else
  \newcommand\contentsname{Table of contents}
\fi
\ifdefined\listfigurename
  \renewcommand*\listfigurename{List of Figures}
\else
  \newcommand\listfigurename{List of Figures}
\fi
\ifdefined\listtablename
  \renewcommand*\listtablename{List of Tables}
\else
  \newcommand\listtablename{List of Tables}
\fi
\ifdefined\figurename
  \renewcommand*\figurename{Figure}
\else
  \newcommand\figurename{Figure}
\fi
\ifdefined\tablename
  \renewcommand*\tablename{Table}
\else
  \newcommand\tablename{Table}
\fi
}
\@ifpackageloaded{float}{}{\usepackage{float}}
\floatstyle{ruled}
\@ifundefined{c@chapter}{\newfloat{codelisting}{h}{lop}}{\newfloat{codelisting}{h}{lop}[chapter]}
\floatname{codelisting}{Listing}

\makeatother
\makeatletter
\@ifpackageloaded{caption}{}{\usepackage{caption}}
\@ifpackageloaded{subcaption}{}{\usepackage{subcaption}}
\makeatother

\ifLuaTeX
  \usepackage{selnolig}  
\fi
\usepackage[]{natbib}
\usepackage{bookmark}

\IfFileExists{xurl.sty}{\usepackage{xurl}}{} 
\hypersetup{
  pdftitle={Hyperbolic Latent Position Models for Hierarchical Bipartite Data},
  pdfauthor={Kisung You},
  pdfsubject={arXiv version with supplementary material},
  pdfkeywords={Gromov product, identifiability, item response theory, posterior contraction, variational inference},
  colorlinks=true,
  linkcolor={blue},
  filecolor={Maroon},
  citecolor={Blue},
  urlcolor={Blue},
  pdfcreator={LaTeX}}

\newtheorem{theorem}{Theorem}[section]
\newtheorem{proposition}{Proposition}[section]
\newtheorem{corollary}{Corollary}[section]
\newtheorem{lemma}{Lemma}[section]
\theoremstyle{definition}

\newtheorem{assumption}{Assumption}[section]
\theoremstyle{plain}
\newtheorem*{mainthmequivalence}{Theorem 4.1}
\newtheorem*{mainthmpopulation}{Theorem 4.2}
\newtheorem*{mainpropcurvature}{Proposition 4.1}
\newtheorem*{mainpropcollapse}{Proposition 4.2}
\newtheorem*{mainproptree}{Proposition 4.3}
\newtheorem*{mainthmlocalid}{Theorem 4.3}
\newtheorem*{mainthmcontraction}{Theorem 6.1}
\newtheorem*{mainpropobserved}{Proposition 6.1}
\newtheorem*{maincorsignal}{Corollary 6.1}
\newtheorem*{mainthmvariational}{Theorem 6.2}

\newcommand{\R}{\mathbb R}
\newcommand{\Hh}{\mathbb H}
\newcommand{\E}{\mathbb E}
\newcommand{\KL}{\operatorname{KL}}
\newcommand{\Renyi}{\operatorname{D}_{2}}
\newcommand{\TN}{\operatorname{TN}}
\newcommand{\PG}{\operatorname{PG}}
\newcommand{\Bern}{\operatorname{Bernoulli}}
\newcommand{\Normal}{\mathcal N}
\newcommand{\calA}{\mathcal A}
\newcommand{\calC}{\mathcal C}

\newcommand{\calX}{\mathcal X}
\newcommand{\calQ}{\mathcal Q}

\newcommand{\one}{\mathbf 1}
\newcommand{\projA}{\mathsf P_{\calA^\perp}}
\newcommand{\bD}{\mathbf D}
\newcommand{\bG}{\mathbf G}
\newcommand{\bP}{\mathbf P}
\newcommand{\bM}{\mathbf M}
\newcommand{\dd}{\mathrm d}
\newcommand{\eps}{\varepsilon}
\newcommand{\argmin}{\operatorname*{arg\,min}}
\newcommand{\cover}{\mathfrak N}
\newcommand{\hBer}{h_{\mathrm{Ber}}}

\DeclareMathOperator{\arcosh}{arcosh}
\DeclareMathOperator{\Exp}{Exp}
\DeclareMathOperator{\Log}{Log}
\DeclareMathOperator{\LCA}{LCA}
\DeclareMathOperator{\clip}{clip}
\DeclareMathOperator{\logit}{logit}

\newcommand{\anon}{1}

\begin{document}

\def\spacingset#1{\renewcommand{\baselinestretch}%
{#1}\small\normalsize} \spacingset{1}


\if1\anon
{
  \title{\bf Hyperbolic Latent Position Models for Hierarchical Bipartite Data}
  \author{Kisung You\textsuperscript{1,2,3}
  \\[0.5em]
  \textsuperscript{1}Department of Mathematics, Baruch College\\
  \textsuperscript{2}Department of Mathematics, CUNY Graduate Center\\
  \textsuperscript{3}Department of Psychiatry, Yale University}
  \maketitle
} \fi

\if0\anon
{
  \bigskip
  \bigskip
  \bigskip
  \begin{center}
    {\LARGE\bf Hyperbolic Latent Position Models for Hierarchical Bipartite Data}
\end{center}
  \medskip
} \fi

\bigskip
\begin{abstract}
Binary bipartite data often combine row and column heterogeneity with hierarchical interaction, as when students answer exercises organized by prerequisites or legislators vote on nested policy areas.
Fixed-dimensional Euclidean volume grows only polynomially with radius, whereas branching hierarchies expand exponentially.
Hyperbolic space matches this growth, and its rooted Gromov product measures shared ancestry.
We propose the \emph{HypErbolic Latent Position model for bipartite Interaction data} (HELPI), which places both object types in hyperbolic space while separating geometric interaction from additive propensities.
With unrestricted main effects, a hyperbolic-distance predictor is likelihood-equivalent to a rooted Gromov-product predictor.
Radial depth is absorbed by additive terms, leaving a root-invariant projected hierarchy signal as the identified geometric target.
We establish identification and posterior contraction results and develop augmented variational procedures for binary outcomes.
Simulations support recovery of hierarchical interaction and clarify weak-branch regimes.
In Junyi Academy data, curriculum anchors yield an interpretable geometry with temporal prediction close to item-response benchmarks, while distinguishing attempted exercises from unconditional mastery.
In U.S. House roll calls, flexible HELPI variants improve on additive and two-parameter logistic benchmarks and recover party and policy structure without using party labels during fitting.
\end{abstract}

\noindent%
\vfill

\newpage
\spacingset{1.8} 

\section{Introduction}
\label{sec:introduction}

Bipartite data record interactions between two kinds of objects. Students answer exercises, consumers select products, citizens follow political accounts, and patients receive diagnoses. For a binary response, the observation is a rectangular array
\[
Y=(Y_{ij})\in\{0,1\}^{n\times m},
\]
where $Y_{ij}=1$ indicates an interaction between row object $i$ and column object $j$. Figure~\ref{fig:bipartite} shows the same data as an array and a bipartite graph. The two views emphasize relational profile and cross-type organization, respectively.

\begin{figure}[ht]
\centering
\includegraphics[width=\textwidth]{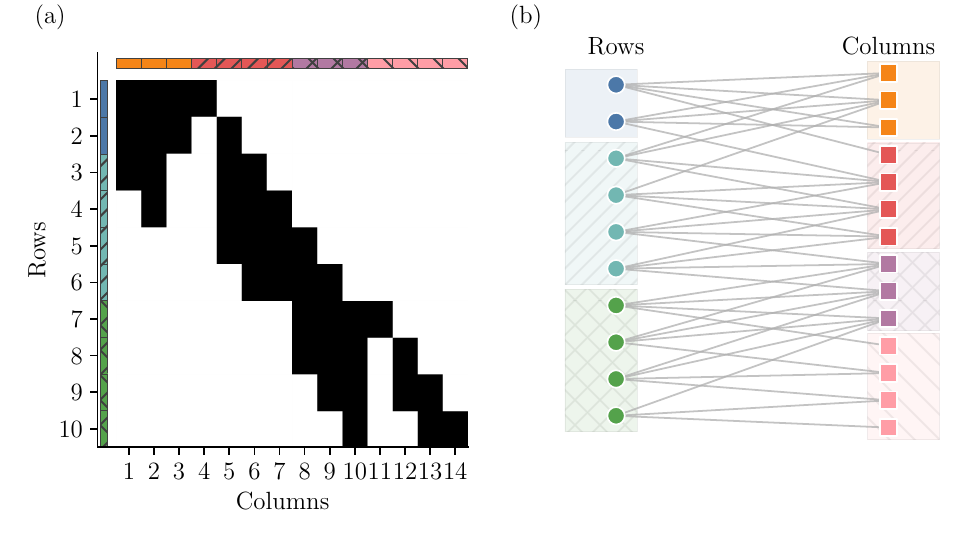}
\caption{A toy bipartite binary dataset in two equivalent forms. Panel (a) is a $10\times14$ row-by-column array with a descending block pattern. Solid cells indicate positive dyads, while filled and hatched marginal bands identify groups. Panel (b) displays the same observations as a graph, with row circles, column squares, and one edge for each solid cell.}
\label{fig:bipartite}
\end{figure}

Many such arrays often reflect hierarchy rather than flat similarity. Educational items follow prerequisite structures, products belong to nested taxonomies, and information sources form topic trees. A hierarchy with branching factor $b$ has $b^h$ descendants at depth $h$. Fixed-dimensional Euclidean volume grows polynomially with radius, whereas hyperbolic volume grows exponentially. This correspondence motivates hyperbolic representations of trees, taxonomies, and complex networks \citep{krioukov_2010_HyperbolicGeometryComplex,nickel_2017_PoincareEmbeddingsLearning,smith_2019_GeometryContinuousLatent}. Figure~\ref{fig:capacity} contrasts the two regimes. Its right panel uses the Poincar\'e disk, while the calculations use the equivalent Lorentz model.

\begin{figure}[t]
\centering
\includegraphics[width=\textwidth]{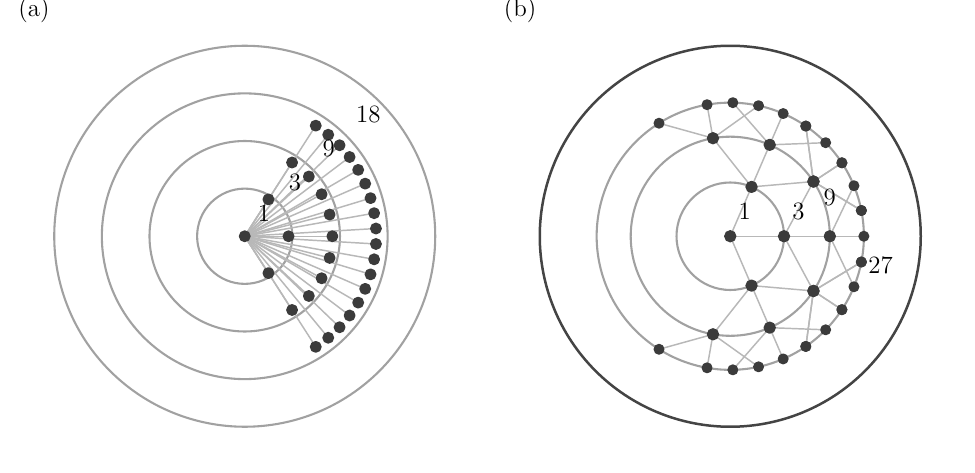}
\caption{Representational capacity in flat and negatively curved spaces. Panel (a) places successively larger node layers on evenly spaced Euclidean circles, illustrating polynomial volume growth. Panel (b) places a three-level branching tree inside a Poincar\'e disk. Its radial spacing increases toward the boundary as available volume grows exponentially with depth.}
\label{fig:capacity}
\end{figure}

Geometry alone, however, does not settle the statistical problem. Bipartite data almost always require row and column main effects: some rows are more active, and some columns are more popular, visible, easy, or prevalent. Once these main effects are admitted, hyperbolic distance no longer identifies absolute depth. Fix a root $o$ and let $z_i,w_j$ be row and column positions. The rooted Gromov product
\[
(z_i\mid w_j)_o=\frac12\{d(o,z_i)+d(o,w_j)-d(z_i,w_j)\}
\]
measures their shared initial path from the root. In a tree it is the depth of the lowest common ancestor. The exact identity
\[
\alpha_i+\beta_j-\gamma d(z_i,w_j)
=
\{\alpha_i-\gamma d(o,z_i)\}
+
\{\beta_j-\gamma d(o,w_j)\}
+
2\gamma(z_i\mid w_j)_o
\]
shows that unrestricted row and column effects absorb the radial terms. The scientifically relevant nonadditive object is therefore shared ancestry, not radial depth by itself.

Built around that observation, we propose HELPI, the \emph{HypErbolic Latent Position model for bipartite Interaction data}. It provides a generalized linear model (GLM) framework that combines geometric interaction with binary, count, ordinal, or other response families, though we primarily focus on the binary case throughout this paper. The binary setting permits transparent logistic and probit computation while keeping the identification argument separate from response-specific details. Other GLMs change the observation model but retain the geometric decomposition.

The exact distance-to-Gromov identity identifies the geometric component that remains after row and column effects are removed. This component is the root-invariant projected hierarchy signal. The accompanying analysis covers curvature-coefficient confounding, ordered single-branch collapse, and local branch resolvability. For computation, we propose  algorithms using P\'olya-Gamma or Albert-Chib augmentation with bounded tangent-space mean-field factors. The theory establishes posterior contraction for the fitted probability matrix and the identified hierarchy signal. Simulations and real-data examples assess prediction, interpretation, and scope across several interaction structures.

The rest of this paper is organized as follows. Section~\ref{sec:background} reviews related work and introduces the required hyperbolic geometry. Sections~\ref{sec:model} through \ref{sec:theory} develop the model, identification analysis, computation, and asymptotic theory. Section~\ref{sec:experiments} reports simulations and real-data applications, and Section~\ref{sec:discussion} concludes. All proofs are contained in the supplementary material, which also gives complete empirical designs, diagnostics, and additional results. Additionally, the code and scripts to reproduce all experiments are available in a publicly accessible repository at \url{https://github.com/kisungyou/HELPI}.

\section{Background}\label{sec:background}

\subsection{Related work}\label{sec:related}

Binary bipartite latent models are established in political methodology, psychometrics, and network statistics. Ideal-point models and spatial-following models infer positions after controlling for row and column heterogeneity \citep{clinton_2004_StatisticalAnalysisRoll,barbera_2015_BirdsSameFeather}. Rasch and two-parameter logistic models separate respondent and item effects \citep{birnbaum_1968_LatentTraitModels,rasch_1980_ProbabilisticModelsIntelligence}, while latent-space item-response and latent-surface models add relational structure \citep{mccormick_2015_LatentSurfaceModels,jeon_2021_MappingUnobservedItem}. These approaches motivate separating marginal propensity from cross-type compatibility.

The closest Euclidean alternatives use latent distances or bilinear terms \citep{hoff_2002_LatentSpaceApproaches,hoff_2005_BilinearMixedEffectsModels}. HELPI instead maps hyperbolic distances or Gromov products into a nonlinear predictor family that need not be low rank. Consequently, matrix-completion and graphon rates do not apply without additional analysis \citep{davenport_2014_1BitMatrixCompletion,gao_2015_RateoptimalGraphonEstimation}. Negative curvature has supported random graph models, hierarchical embeddings, recommender systems, and Bayesian network models \citep{krioukov_2010_HyperbolicGeometryComplex,kitsak_2017_LatentGeometryBipartite,nickel_2017_PoincareEmbeddingsLearning,ganea_2018_HyperbolicNeuralNetworks,chamberlain_2019_ScalableHyperbolicRecommender,smith_2019_GeometryContinuousLatent,vinhtran_2020_HyperMLBoostingMetric,papamichalis_2022_LatentSpaceNetwork}. HELPI determines which feature of this geometry remains identifiable after saturated main effects are included. Its binary algorithms combine standard logistic or probit augmentation with bounded tangent-space factors \citep{albert_1993_BayesianAnalysisBinary,jaakkola_2000_BayesianParameterEstimation,polson_2013_BayesianInferenceLogistic,durante_2019_ConditionallyConjugateMeanField,mathieu_2019_ContinuousHierarchicalRepresentations,nagano_2019_WrappedNormalDistribution}.

\subsection{Hyperbolic geometry and rooted hierarchy}
\label{sec:geometry}

We use the Lorentz hyperboloid of sectional curvature $-1$,
\[
\begin{aligned}
\Hh^p&=\{x\in\R^{p+1}:\langle x,x\rangle_{\mathrm L}=-1,\ x_0>0\},\\
\langle x,y\rangle_{\mathrm L}&=-x_0y_0+\sum_{k=1}^p x_ky_k,
\qquad d(x,y)=\arcosh\{-\langle x,y\rangle_{\mathrm L}\}.
\end{aligned}
\]
Fix the root $o=(1,0,\ldots,0)^\top$. For $\tau\in\R^p$ with $\ell=\|\tau\|$, its exponential map is
\[
\Exp_o(\tau)=(\cosh\ell,\ \sinh\ell\,\tau^\top/\ell)^\top,
\]
with $\Exp_o(0)=o$. The norm $\ell$ represents depth, and the direction $\tau/\ell$ represents branch orientation. The rooted Gromov product
\[
G_o(z,w)=(z\mid w)_o=\tfrac12\{d(o,z)+d(o,w)-d(z,w)\}
\]
measures shared travel from the root \citep{bridson_1999_MetricSpacesNonpositive}. For row positions $z_i$ and column positions $w_j$, write $G_{ij}=G_o(z_i,w_j)$ and $\bG=(G_{ij})$.

Define the row-column additive subspace
\[
\calA=\{a\one_m^\top+\one_n b^\top:a\in\R^n,\ b\in\R^m\},
\]
and let $\projA$ be orthogonal projection onto $\calA^\perp$ in Frobenius norm. If the root changes from $o$ to $o'$, then $G_o(z_i,w_j)-G_{o'}(z_i,w_j)$ is a row term plus a column term. Hence $\projA\bG$ is root invariant.

One may use any negative value for the curvature of target latent space. However, we use the fixed curvature value of $-1$ throughout this paper. It will later be shown that the geometry coefficient can offset any change in curvature. If curvature were also estimated, a prior whose form changes under this rescaling could favor one likelihood-equivalent parameter pair over another. That posterior preference would be induced by the prior rather than identified by the data.

\section{Model, observation scheme, and priors}\label{sec:model}

\subsection{Binary HELPI and broader GLM scope}

For a suitable response family and link, HELPI adds a geometric interaction to row and column effects. This paper specializes to binary responses,
\[
Y_{ij}\mid\pi_{ij}\sim\Bern(\pi_{ij}),
\qquad \pi_{ij}=H(\eta_{ij}),
\]
where $H$ is either the logistic distribution function $\Lambda$ or the standard normal distribution function $\Phi_0$. The distance and shared-ancestry predictors are modeled as 
\[
\eta_{ij}=\alpha_i+\beta_j-\gamma d(z_i,w_j)
\quad\text{and}\quad
\eta_{ij}=\alpha_i+\beta_j+\lambda G_o(z_i,w_j),\qquad \gamma,\lambda>0,
\]
respectively.  The main development uses the second form; other GLMs retain the same geometric decomposition.

Let $A_{ij}=\alpha_i+\beta_j$. We note that $A\in\calA$ is identified since the decomposition into $(\alpha,\beta)$ is invariant under $(\alpha_i,\beta_j)\mapsto(\alpha_i+c,\beta_j-c)$. A centering convention may be imposed for computation without changing the model.

\subsection{Observed dyads and exposure}

Let $\Omega$ be the observed-dyad set. Conditional on the parameters and $\Omega$, its entries are independent with likelihood
\[
L_\Omega(\theta)=\prod_{(i,j)\in\Omega}
\pi_{ij}(\theta)^{Y_{ij}}\{1-\pi_{ij}(\theta)\}^{1-Y_{ij}}.
\]
All computational sums are over $\Omega$. Treating $\Omega$ as fixed is appropriate when exposure is fixed or ignorable given observed information. This distinction matters in education because an unattempted exercise is missing rather than incorrect. The Junyi estimand is therefore correctness conditional on an attempt. Estimating unconditional mastery would require a joint exposure-response model. Theorem~\ref{thm:contraction} treats a complete array, and Proposition~\ref{prop:observed} gives the fixed-design observed-dyad result.

\subsection{Priors and bounded-depth model}

Writing $z_i=\Exp_o(\tau_i)$ and $w_j=\Exp_o(\upsilon_j)$, the bounded-depth model uses
\[
\alpha_i\stackrel{\mathrm{ind}}\sim\Normal(0,\sigma_\alpha^2),
\qquad
\beta_j\stackrel{\mathrm{ind}}\sim\Normal(0,\sigma_\beta^2),
\qquad
\lambda\sim\TN_{[0,\infty)}(0,\sigma_\lambda^2),
\]
and multivariate Gaussian tangent priors truncated to $B_{\mathrm E}(R_0)=\{u:\|u\|\le R_0\}$. The asymptotic analysis imposes the same depth bound on the latent positions. We select $R_0$ from external depth information or a prespecified predictive grid and report the fitted boundary fraction. A known hierarchy can fix column positions at anchors $\widetilde w_j$ or center soft priors at their tangent coordinates. Supplementary Section S5 gives the full prior and diagnostic specification.


\section{Identification and hierarchical interpretation}
\label{sec:identification}

\subsection{Projected hierarchy signal and population equivalence}

Let $\bM=(\eta_{ij})$ and $\bP=(H(\eta_{ij}))$. Since both links are strictly increasing, $\bM$ and $\bP$ determine one another at the population level. Define
\[
\mathcal H=\lambda\projA\bG.
\]
We call $\mathcal H$ the {identified hierarchy signal} and use that term consistently below.

\begin{theorem}[Distance-Gromov-product equivalence]
\label{thm:equivalence}
For any root $o$, set
\[
\widetilde\alpha_i=\alpha_i-\gamma d(o,z_i),
\quad
\widetilde\beta_j=\beta_j-\gamma d(o,w_j),
\quad
\lambda=2\gamma.
\]
Then
\[
\alpha_i+\beta_j-\gamma d(z_i,w_j)
=
\widetilde\alpha_i+\widetilde\beta_j+\lambda(z_i\mid w_j)_o.
\]
Conversely, every shared-ancestry predictor with $\lambda>0$ has a distance representation with $\gamma=\lambda/2$ after redefining the main effects.
\end{theorem}

Theorem~\ref{thm:equivalence} shows that latent distance and shared ancestry give the same predictors after radial terms enter the main effects. Identification therefore depends on the component that remains after additive effects are removed. The following theorem characterizes that quotient parameter.

\begin{theorem}[Population equivalence modulo additive effects]
\label{thm:population}
Let $(\lambda,\bG)$ and $(\lambda',\bG')$ be two hierarchical parameter pairs. There exist additive matrices $A,A'\in\calA$ such that
\[
A+\lambda\bG=A'+\lambda'\bG'
\]
if and only if
\[
\lambda'\bG'-\lambda\bG\in\calA.
\]
Equivalently, the two hierarchical pairs are observationally equivalent under unrestricted main effects if and only if
\[
\lambda\projA\bG=\lambda'\projA\bG'.
\]
The same identified hierarchy signal satisfies
\[
\mathcal H=-\gamma\projA\bD
\]
in the distance representation and is invariant to the choice of root.
\end{theorem}

If the additive terms are fixed, equality of the two predictors requires the stronger identity $A-A'=\lambda'\bG'-\lambda\bG$. Theorem~\ref{thm:population} instead characterizes hierarchical equivalence when the additive terms may vary freely.

\begin{proposition}[Curvature-coefficient confounding in the likelihood]
\label{prop:curvature}
Let the same differentiable manifold carry the curvature-$-1$ metric $\mathfrak g_1$ and the rescaled metric $\mathfrak g_{\kappa_c}=\kappa_c^{-1}\mathfrak g_1$, whose sectional curvature is $-\kappa_c$. Under the identity map between these two metric spaces,
\[
d_{\kappa_c}=\kappa_c^{-1/2}d_1,
\qquad
(\cdot\mid\cdot)_{o,\kappa_c}=\kappa_c^{-1/2}(\cdot\mid\cdot)_{o,1}.
\]
Thus the likelihood depends on $(\gamma,\kappa_c)$ or $(\lambda,\kappa_c)$ only through $\gamma/\sqrt{\kappa_c}$ or $\lambda/\sqrt{\kappa_c}$. A prior whose form changes under this rescaling can favor one equivalent curvature-coefficient pair over another. Such a posterior preference is induced by the prior and is not identified by the data.
\end{proposition}

Fixing curvature removes the scale alias before branch structure is considered. Proposition~\ref{prop:collapse} then identifies an exact configuration where the projected branch signal disappears.

\begin{proposition}[Ordered single-branch collapse]
\label{prop:collapse}
Suppose all $z_i,w_j$ lie on one geodesic ray from $o$, with depths $\ell_i^{(z)}=d(o,z_i)$ and $\ell_j^{(w)}=d(o,w_j)$. If either $\min_i \ell_i^{(z)}>\max_j \ell_j^{(w)}$ or $\min_j \ell_j^{(w)}>\max_i \ell_i^{(z)}$, then $\bD\in\calA$ and $\projA\bG=0$. The geometric contribution is therefore observationally equivalent to main effects alone.
\end{proposition}

Proposition~\ref{prop:collapse} describes a boundary where ordered radial variation is entirely additive. Under genuine branching, the projected Gromov matrix has a direct common-ancestry interpretation. For an exact rooted tree metric, this interpretation reduces to lowest-common-ancestor depth.

\begin{proposition}[Tree interpretation]
\label{prop:lca}
For an exact rooted tree metric,
\[
G_{ij}=\operatorname{depth}\{\LCA(z_i,w_j)\}.
\]
Hence $\mathcal H$ is the projected matrix of lowest-common-ancestor depths multiplied by $\lambda$.
\end{proposition}

An arbitrary finite tree metric need not embed isometrically in a fixed-dimensional hyperbolic space. Proposition~\ref{prop:lca} is exact for the tree metric itself. For a hyperbolic embedding, it describes the target approached as embedding distortion decreases. In a $\delta$-hyperbolic approximation, Gromov products recover common-ancestor depth up to distortion controlled by $\delta$ \citep{bridson_1999_MetricSpacesNonpositive}.

\subsection{Constructive branch resolvability}

On the collision-free tangent domain, define the analytic position-to-Gromov map
\[
\Psi(\tau_{1:n},\upsilon_{1:m})=
\left[G_o\{\Exp_o(\tau_i),\Exp_o(\upsilon_j)\}\right]_{ij}.
\]
Within each connected component, consider a generic point $x_0$ where $D\Psi$ attains its maximal rank. The local image $\mathfrak S_G$ is then a smooth manifold with tangent space $T_{\bG}\mathfrak S_G=\operatorname{ran}D\Psi_{x_0}$. Root-fixing rotations do not change the Gromov matrix, so their tangent directions lie in the kernel of $D\Psi_{x_0}$. Let $E_{\calA}$ and $E_G$ be orthonormal bases for $\calA$ and $T_{\bG}\mathfrak S_G$, respectively, and define
\[
J_{\rm br}=\left[E_{\calA},\operatorname{vec}(\bG),\lambda E_G\right].
\]
Its columns represent additive directions, changes in the hierarchy coefficient, and feasible changes in branch shape. Supplementary Section S3 establishes the analytic genericity statement and the corresponding dimension bounds.

\begin{theorem}[Local branch identification]
\label{thm:localid}
At a maximal-rank generating configuration, suppose $J_{\rm br}$ has full column rank. Then, in a neighborhood of the reference point, the map
\[
(A,\lambda,\bG)\mapsto A+\lambda\bG
\]
is one-to-one on $\calA\times(0,\infty)\times\mathfrak S_G$. Thus $(A,\lambda,\bG)$ is jointly locally identified; only the representation $A_{ij}=\alpha_i+\beta_j$ retains the one-dimensional common-shift alias. If $J_{\rm br}$ is rank deficient, a nonzero first-order change in additive effects, hierarchy scale, or hierarchy shape leaves the predictor unchanged to first order.
\end{theorem}

In computation, $E_G$ is an orthonormal basis for the range of $D\Psi$ evaluated at the fitted tangent coordinates. The smallest singular value of $J_{\rm br}$ therefore measures local resolvability. This diagnostic requires fixed negative curvature. In an ordinary Euclidean latent-distance model, rescaling all positions rescales every distance. The resulting scale direction lies in the tangent space of feasible distance matrices, so a separate distance coefficient is not locally identifiable. Hyperbolic space at fixed curvature has no global dilation that preserves the model class. Full rank of $J_{\rm br}$ can therefore hold. This argument connects Proposition~\ref{prop:curvature} to Theorem~\ref{thm:localid}.
\section{Computation}
\label{sec:computation}

The practical method combines exact scalar updates with stochastic geometry updates. All likelihood sums are restricted to the observed set $\Omega$.

\subsection{Variational approximation and local augmentation}

Let $z_i=\Exp_o(\tau_i)$ and $w_j=\Exp_o(\upsilon_j)$. The bounded-depth approximation uses Gaussian tangent factors truncated to $B_{\mathrm E}(R_0)$, Gaussian main-effect factors, and a positive-truncated Gaussian factor for $\lambda$ \citep{wainwright_2008_GraphicalModelsExponential,blei_2017_VariationalInferenceReview}.

For each observed dyad, the algorithm tracks the predictor moments
\[
\mu_{ij}^{\eta}=\E_q(\eta_{ij}),
\qquad
\nu_{ij}^{\eta}=\E_q(\eta_{ij}^2).
\]
The second quantity is a raw second moment. Logistic inference uses P\'olya-Gamma augmentation, which recovers the Jaakkola-Jordan local variational update \citep{jaakkola_2000_BayesianParameterEstimation,polson_2013_BayesianInferenceLogistic,durante_2019_ConditionallyConjugateMeanField}. Probit inference instead uses the optimal sign-truncated Gaussian factor from the Albert-Chib representation. Supplementary Section S5 gives the factor definitions and numbered update equations; the parenthetical equation numbers in Algorithm~\ref{alg:helpi} refer to that section.

\subsection{Block updates and implementation}

\begin{algorithm}[tbp]
\caption{Blockwise mean-field variational inference for binary HELPI}
\label{alg:helpi}
\begin{algorithmic}[1]
\Require Observed set $\Omega$, link $H$, depth bound $R_0$, and independent restarts
\State Initialize the main effects, $q_\lambda$, and the bounded tangent factors.
\Repeat
\State Update the predictor moments by (S1a) and (S1b).
\State Update local factors by (S2) for logit or (S3) for probit.
\State Update $q_\alpha$, $q_\beta$, and $q_\lambda$ by (S4a) through (S4c).
\State Update the tangent factors by the projected step (S5).
\State Recenter and refresh convergence summaries by (S6a) and (S6b).
\Until{the smoothed objective and invariant summaries stabilize}
\State \Return the restart with the largest high-accuracy objective estimate.
\end{algorithmic}
\end{algorithm}

Conditional on the geometry moments, the scalar factors have exact coordinate updates. Each nonconjugate tangent block receives one projected stochastic-gradient step. The default gradient estimator uses a score-function estimator with control variates. Hyperbolic gradient calculations are reviewed by \citet{wilson_2018_GradientDescentHyperbolic}, while \citet{mathieu_2019_ContinuousHierarchicalRepresentations} and \citet{nagano_2019_WrappedNormalDistribution} develop wrapped distributions and their gradients.

The tangent means remain in $B_{\mathrm E}(R_0-\delta_R)$, and covariance eigenvalues remain in a prespecified compact interval. Stopping rules use smoothed or periodically refined objective estimates because Monte Carlo geometry updates are noisy. We use multiple starts to reduce sensitivity to nonconvexity. When the branch-resolvability diagnostic from Theorem~\ref{thm:localid} is small, we emphasize $\bP$ and $\mathcal H$ rather than interpreting $\lambda$ separately.

For dense data, one geometry sweep costs $O(R_{\rm mc}nmp)$. With partial observation it costs $O(R_{\rm mc}|\Omega|p)$. This computational scaling does not itself justify inference under informative exposure; it only describes the conditional observed-dyad likelihood.

\section{Theory}\label{sec:theory}

\subsection{Posterior contraction and hierarchy-signal risk}

For the fully observed array, write
\[
P_\theta=\bigl(H\{\eta_{ij}(\theta)\}\bigr)_{i,j},
\qquad N=nm,
\qquad d_N=(n+m)(p+1)+1.
\]
Both $n$ and $m$ grow, $p$ is fixed, and $d_N\log N=o(N)$. Supplementary Section S4 gives the sieves, entropy bounds, and tests.

\begin{assumption}[Truth and priors]
\label{ass:theory}
The true main effects and $\lambda_0\ge0$ are uniformly bounded, and the true tangent coordinates lie strictly inside $B_{\mathrm E}(R_0)$. Scalar priors are positive on their supports, have sub-Gaussian tails, and have at most quadratic negative log density. Tangent-prior densities are continuous and bounded above and below by positive constants on $B_{\mathrm E}(R_0)$.
\end{assumption}

This formulation includes the no-hierarchy case $\lambda_0=0$. Define
\[
\eps_N^2=K_0\frac{d_N\log N}{N}
\]
for a sufficiently large constant $K_0$, and the average Bernoulli Hellinger metric
\[
h_N^2(P,Q)=\frac1N\sum_{i,j}
\left[(\sqrt{P_{ij}}-\sqrt{Q_{ij}})^2+
(\sqrt{1-P_{ij}}-\sqrt{1-Q_{ij}})^2\right].
\]

\begin{theorem}[Dense-array posterior contraction]
\label{thm:contraction}
Under Assumption~\ref{ass:theory}, for either the logistic or probit link and every deterministic sequence $a_N\to\infty$,
\[
\Pi\{\theta:h_N(P_\theta,P_0)>a_N\eps_N\mid Y\}\to0
\]
in $P_0$-probability.
\end{theorem}

Theorem~\ref{thm:contraction} concerns average predictive error over the complete array. Proposition~\ref{prop:observed} adapts this rate to a deterministic observed design through its effective sample size and active-node dimension.

\begin{proposition}[Conditional contraction for an observed dyad set]
\label{prop:observed}
Let $\Omega_N$ be a deterministic observed set of size $M_N$, and let $d_{\Omega,N}$ replace $n,m$ in $d_N$ by the numbers of active rows and columns. Define
\[
h_{\Omega_N}^2(P,Q)=\frac1{M_N}\sum_{(i,j)\in\Omega_N}h_{\rm Ber}^2(P_{ij},Q_{ij}).
\]
If $d_{\Omega,N}\log M_N=o(M_N)$, the posterior based on $L_{\Omega_N}$ contracts in $h_{\Omega_N}$ at rate
\[
\eps_{\Omega,N}^2=K_0d_{\Omega,N}\log(M_N)/M_N.
\]
The same conditional conclusion holds with high probability under response-ignorable Bernoulli sampling when its expected sample size satisfies the analogous dimension condition.
\end{proposition}

Proposition~\ref{prop:observed} controls only prediction on the observed design. A full-array version of Corollary~\ref{cor:signal} would require an additional condition linking error on observed dyads to the complete-array Frobenius norm. Informative exposure remains outside the theory.

For a clip level $B_{\rm clip}\ge\|\bM_0\|_\infty$, define
\[
g_{B_{\rm clip}}(t)=\clip\{H^{-1}(t),[-B_{\rm clip},B_{\rm clip}]\},
\qquad
\mathcal H_{B_{\rm clip}}(P)=\projA[g_{B_{\rm clip}}(P)].
\]
Then $\mathcal H_{B_{\rm clip}}(P_0)=\lambda_0\projA\bG_0$.

\begin{corollary}[Contraction of the identified hierarchy signal]
\label{cor:signal}
There is a finite constant $C_{B_{\rm clip}}$ such that, for every $a_N\to\infty$,
\[
\Pi\left\{\theta:
N^{-1/2}\|\mathcal H_{B_{\rm clip}}(P_\theta)-\mathcal H_{B_{\rm clip}}(P_0)\|_F
>a_N C_{B_{\rm clip}}\eps_N\mid Y\right\}\to0
\]
in $P_0$-probability.
\end{corollary}

The clipping level determines the reported functional. In applications, $B_{\rm clip}$ should be prespecified, the fraction of fitted predictors affected by clipping should be reported, and sensitivity over a modest grid should be examined.

\subsection{Exact mean-field variational risk}

Let $\calQ_N^{\rm mf}$ contain the product factors from Section~\ref{sec:computation}. For each $N$, their means and covariance eigenvalues lie in the compact parameter set specified in Supplementary Section S4. These bounds expand with $N$, retain factors concentrated at the truth, and ensure that a variational minimizer exists.

\begin{theorem}[Predictive risk of the exact mean-field optimum]
\label{thm:vb}
A global reverse-KL minimizer
\[
\widehat Q_N\in\argmin_{Q\in\calQ_N^{\rm mf}}
\KL\{Q\,\|\,\Pi(\cdot\mid Y)\}
\]
exists. Under Assumption~\ref{ass:theory},
\[
\E_{P_0}\int h_N^2(P_\theta,P_0)\,\widehat Q_N(\dd\theta)
\le C\eps_N^2
\]
for a constant $C<\infty$. Moreover,
\[
\E_{P_0}\int
N^{-1}\|\mathcal H_{B_{\rm clip}}(P_\theta)-\mathcal H_{B_{\rm clip}}(P_0)\|_F^2
\,\widehat Q_N(\dd\theta)
\le C C_{B_{\rm clip}}^2\eps_N^2.
\]
\end{theorem}

This theorem concerns the exact global reverse-KL minimizer over $\calQ_N^{\rm mf}$. It does not establish the same risk bound for a local stationary point or a tempered variational target \citep{alquier_2020_ConcentrationTemperedPosteriors}.

\section{Empirical studies}
\label{sec:experiments}

Two simulations examine hierarchy recovery and branch collapse. Two real-data examples then study educational responses and congressional voting using held-out data. Supplementary Simulations 3 and 4 examine complementary interaction structures and analyst choices. We report log score, Brier score, area under the receiver operating characteristic curve (AUC), root mean squared error (RMSE), and calibration. For strictly proper scoring rules, we use the log and Brier scores \citep{gneiting_2007_StrictlyProperScoring}.

Every simulation cell contains 100 independent replicates, and competing methods use the same training and test dyads. Each split assigns 80\% of dyads to training while retaining every active row and column. Plotted intervals are means plus or minus 1.96 standard errors. Supplementary Section S6 gives the complete designs, paired contrasts, calibration results, diagnostics, and execution records.

\subsection{Simulation 1: Hierarchy-signal recovery}
\label{sec:balanced-experiment}

Let \(r_i,c_j\in\{0,\ldots,b-1\}^d\) denote the root-to-leaf paths assigned to row \(i\) and column \(j\). Their lowest-common-ancestor depth is
\[
L_{ij}=\sum_{\ell=1}^d\prod_{s=1}^{\ell}\mathbf 1(r_{is}=c_{js}),
\qquad
\widetilde{\mathbf L}=\projA\mathbf L.
\]
For each replicate, \(a_i\) and \(b_j^\star\) are independent \(\Normal(0,0.45^2)\) variables. We set \(\alpha_i=a_i-\bar a\) and \(\beta_j=b_j^\star-\bar b^\star\), then generate
\[
Y_{ij}\mid P_{ij}\sim\Bern(P_{ij}),
\qquad
\logit(P_{ij})=\mu_\rho+\alpha_i+\beta_j+\lambda\widetilde L_{ij}.
\]
The intercept \(\mu_\rho\) is solved numerically so that \(n^{-1}m^{-1}\sum_{i,j}P_{ij}=\rho\). This construction makes \(\lambda\widetilde{\mathbf L}\) the true identified hierarchy signal.

The design uses \(b\in\{2,3,4\}\), \(d\in\{3,4\}\), \(n=m\in\{100,200,400\}\), \(\lambda\in\{0.25,0.5,1\}\), and \(\rho\in\{0.1,0.3,0.5\}\). A deterministic pairwise cover selects 14 of the 162 factorial cells, with \(p=4\) throughout. The goal is to evaluate held-out prediction and recovery of \(\lambda\widetilde{\mathbf L}\). HELPI is compared with main effects, Euclidean distance, spherical position, and additive-and-multiplicative effects (AME).

Figure~\ref{fig:balanced-exemplar} shows binary networks generated under three cells of this design. To keep individual dyads visible, each panel displays a path-stratified \(26\times26\) induced subnetwork from a generated \(100\times100\) array. These draws illustrate the response structure and are not included in the Monte Carlo summaries.

\begin{figure}[tbp]
\centering
\includegraphics[width=\textwidth]{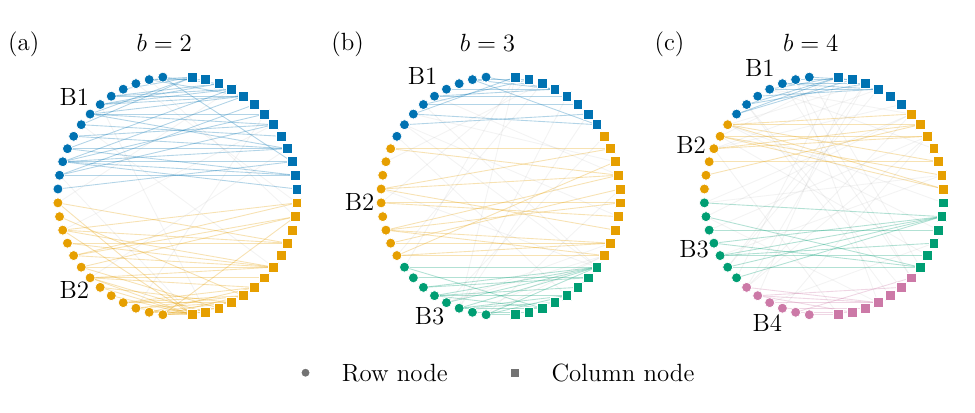}
\caption{Generated-network exemplars for Simulation~1. Panels (a), (b), and (c) use \(b=2,3,4\), respectively, with \(d=3\), \(n=m=100\), \(\lambda=1\), \(\rho=0.1\), and \(p=4\). Each panel shows a path-stratified \(26\times26\) induced subnetwork, with positive dyads drawn as edges. Row circles and column squares occupy opposite semicircles. Labels B1 through B4 identify contiguous first-level branches, and darker edges join nodes in the same branch.}
\label{fig:balanced-exemplar}
\end{figure}

\begin{figure}[tbp]
\centering
\includegraphics[width=\textwidth]{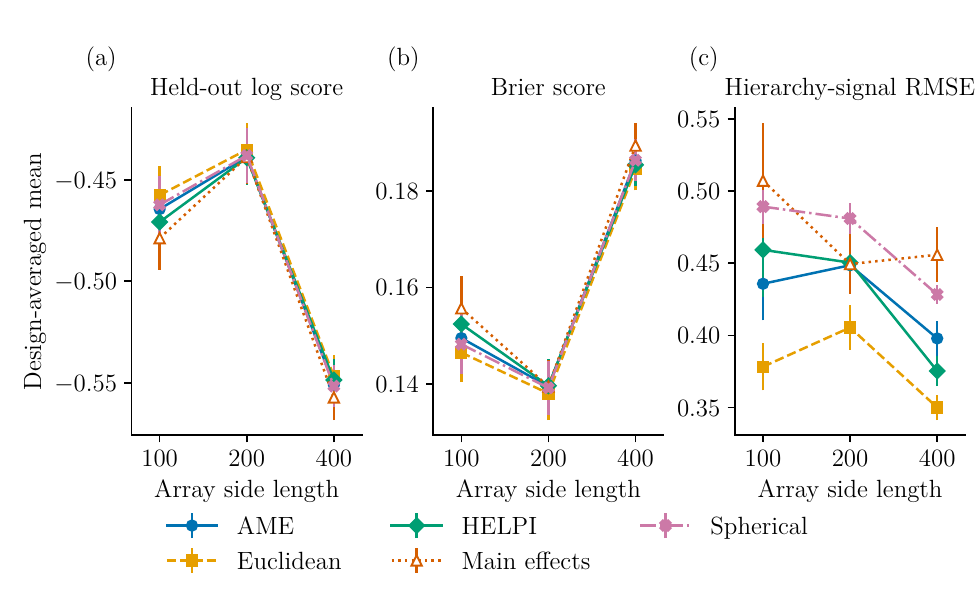}
\caption{Balanced-hierarchy simulation results. Three aligned line plots report design-averaged log score, Brier score, and hierarchy-signal RMSE across array side lengths 100, 200, and 400. Points are means, and bars are 95\% Monte Carlo intervals. Marker shapes and line styles distinguish the five methods. Factor compositions differ across sizes because each size uses its assigned pairwise-covering cells.}
\label{fig:simulation-performance}
\end{figure}

Figure~\ref{fig:simulation-performance} summarizes results by array size, while Supplementary Tables S3 and S4 give cell summaries and paired contrasts. Across the selected design, HELPI improves every reported metric relative to main effects, and the paired intervals for log score, Brier score, and AUC exclude zero. The hierarchy-signal results also show recovery beyond the additive benchmark. The Euclidean model has the best average value for several metrics. In these selected cells, it approximates the generated hierarchy signal well despite using a different geometry.

\subsection{Simulation 2: Branch collapse and weak separation}
\label{sec:collapse-experiment}

Let \(s\in\{0,0.01,0.03,0.1,0.3,0.6,1\}\) denote angular branch spread. Row and column tangent coordinates are
\[
\tau_i=D_i u_i(s),
\qquad
\upsilon_j=E_j v_j(s),
\qquad
D_i\sim\operatorname{Unif}(3,4),
\qquad
E_j\sim\operatorname{Unif}(0.6,1.6).
\]
The first two coordinates of \(u_i(s)\) are proportional to \(\{\cos(s\epsilon_i),\sin(s\epsilon_i)\}\), where \(\epsilon_i\sim\Normal(0,1)\). The directions \(v_j(s)\) are generated independently in the same way. Remaining coordinates receive independent \(\Normal(0,0.01s^2)\) perturbations before normalization. With \(G_{ij}=G_o\{\Exp_o(\tau_i),\Exp_o(\upsilon_j)\}\), responses follow
\[
\logit(P_{ij})=\mu_{0.30}+\alpha_i+\beta_j
+0.75(\projA\bG)_{ij},
\qquad n=m=80,\quad p=4,
\]
using the centered Gaussian main effects from Simulation~1. At \(s=0\), every point lies on one ordered ray and \(\projA\bG=0\) by Proposition~\ref{prop:collapse}.

The goal is to compare the fitted quotient singular value \(\widehat\sigma_{\min}(J_{\rm br})\) with the across-restart standard deviation of \(\widehat\lambda\). The fitted singular value measures local separation among additive, coefficient, and branch-shape directions. The across-restart standard deviation measures optimization sensitivity in the separately reported coefficient. Figure~\ref{fig:collapse-exemplar} shows generated binary networks at three branch spreads, and Figure~\ref{fig:collapse-resolvability} reports the replicate-level and spread-level diagnostic relationship.

\begin{figure}[tbp]
\centering
\includegraphics[width=\textwidth]{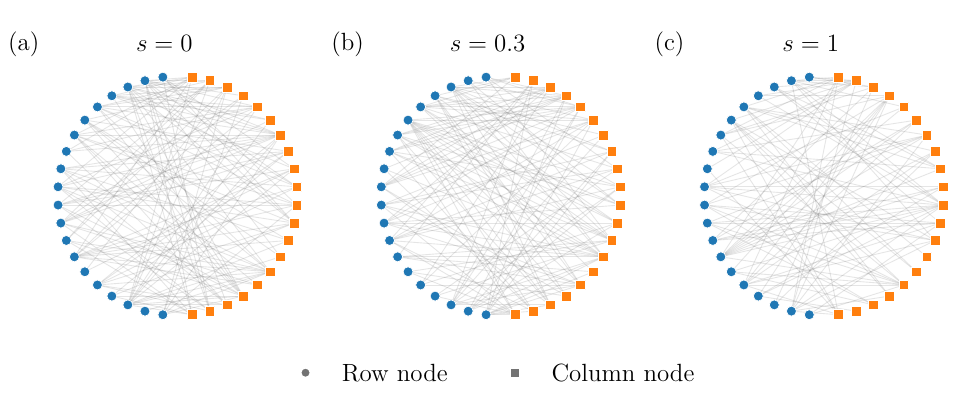}
\caption{Generated-network exemplars for Simulation~2. Panels (a), (b), and (c) use branch spreads \(s=0,0.3,1\), respectively, with \(n=m=80\), \(p=4\), \(\lambda=0.75\), and \(\rho=0.3\). Each panel is an angle- and depth-stratified \(20\times20\) induced subnetwork. Row circles occupy the left semicircle, and column squares occupy the right semicircle. Connecting lines represent realized positive dyads.}
\label{fig:collapse-exemplar}
\end{figure}

\begin{figure}[tbp]
\centering
\includegraphics[width=\textwidth]{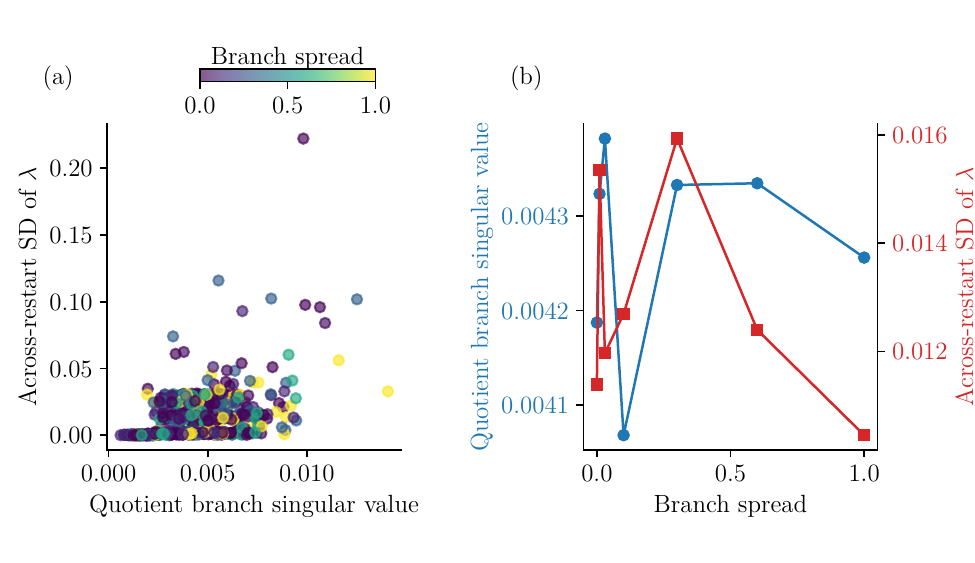}
\caption{Collapse and weak-branch results. Panel (a) is a shaded scatterplot relating replicate-level fitted quotient singular values to across-restart standard deviations of \(\lambda\). Ordered shading indicates branch spread. Panel (b) traces their spread-specific means, with circles using the left axis and squares using the right axis.}
\label{fig:collapse-resolvability}
\end{figure}

At $s=0$, the population hierarchy signal collapses exactly by construction. Across finite-sample fits, however, the fitted diagnostic varies little over the spread grid. Its association with restart variation is positive rather than the prespecified negative ordering. We therefore interpret \(\widehat\sigma_{\min}(J_{\rm br})\) jointly with restart and signal summaries instead of as a stand-alone diagnostic. Supplementary Table S5 gives spread-specific results. Supplementary Sections S6.3 and S6.4 report the complementary-structure and analyst-choice simulations.

\subsection{Real data example 1: Learning from Junyi Academy}
\label{sec:junyi-experiment}

The analysis uses the 2015 PSLC DataShop Junyi Academy release, dataset 1198, and exercise metadata \citep{chang_2015_ModelingExerciseRelationships,junyiacademy_2015_JunyiAcademyMath}. For each observed student-exercise pair \((i,j)\in\Omega\), \(Y_{ij}=1\) denotes a correct first attempt. The conditional response model is
\[
Y_{ij}\mid (i,j)\in\Omega\sim\Bern(P_{ij}),
\qquad
\logit(P_{ij})=\alpha_i+\beta_j+\lambda(z_i\mid w_j)_o.
\]
Streaming preprocessing retains 2,289,789 unique first-attempt pairs from 25,925,992 log rows. A deterministic core contains 2,000 students, 500 exercises, and 317,091 observed pairs. Within each student, the earliest 80\% of attempts form training data and the remainder form testing data. Exercises absent from training are excluded, and the estimand is correctness conditional on an observed attempt.

Each exercise has a path \(h_j=(\text{area}_j,\text{topic}_j,\text{exercise}_j)\), which is embedded as a tangent anchor \(\widetilde\upsilon_j\). Fixed HELPI sets \(\upsilon_j=\widetilde\upsilon_j\), soft HELPI centers its tangent prior at \(\widetilde\upsilon_j\) with scale \(0.35\), and unanchored HELPI uses a zero-centered prior. Comparisons include Rasch, two-parameter logistic (2PL), Euclidean, spherical, and AME models. Every nonconvex fit uses five starts. The goal is to evaluate temporal prediction and determine whether the curriculum taxonomy yields an interpretable geometry.

\begin{figure}[tbp]
\centering
\includegraphics[width=\textwidth]{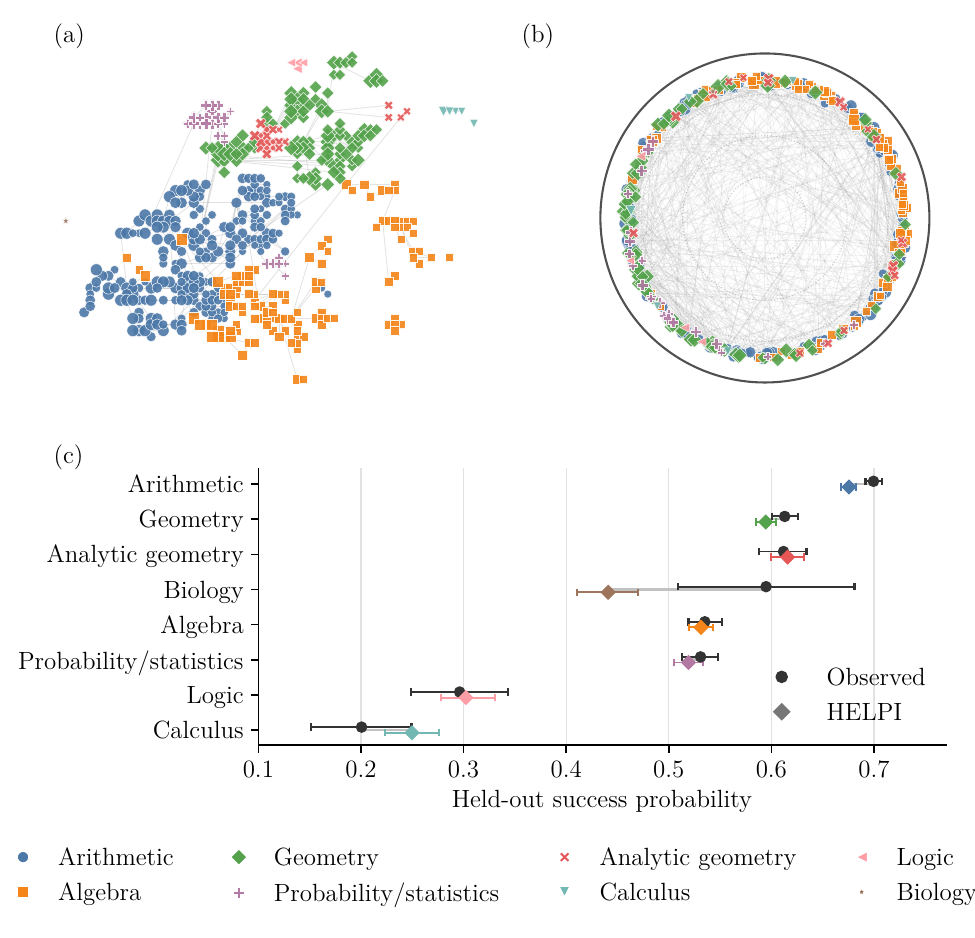}
\caption{Junyi curriculum structure and soft-anchor HELPI fit. Panel (a) arranges exercises and prerequisite links by curriculum coordinates. Marker shape and fill identify curricular area, while node size indicates retained activity. Panel (b) projects the same exercises and links into a two-dimensional Poincar\'e disk. Panel (c) compares observed and fitted mean success probabilities by area, with horizontal 95\% student-cluster bootstrap intervals.}
\label{fig:junyi-network}
\end{figure}

On the primary split, 2PL has log score \(-0.5090\), Brier score \(0.1686\), and AUC \(0.8075\). Soft HELPI records \(-0.5125\), \(0.1701\), and \(0.8038\), respectively, and its log score is within $0.0067$ of the best model in each curricular window reported in Supplementary Table S10. Figure~\ref{fig:junyi-network} shows curriculum organization in the fitted geometry. Area-level fitted probabilities track observed success except in the sparse biology area. The taxonomy therefore provides an interpretable summary with temporal prediction near established item-response models. HELPI fits have zero depth-boundary fraction. Their local diagnostic matrices are rank deficient, so we interpret fitted probabilities and the identified hierarchy signal rather than \(\lambda\) alone.

To examine exposure, define \(R_{ij}=\mathbf 1\{(i,j)\in\Omega\}\). A reproducible case-control analysis pairs each held-out observed dyad with two unobserved dyads. Its AUC shows that attempted and sampled unattempted pairs differ systematically. Because the control prevalence is fixed, these probabilities do not estimate population exposure without correction \citep{prentice_1979_LogisticDiseaseIncidence}. HELPI therefore describes correctness after an attempt rather than unconditional mastery. Supplementary Tables S9 through S11 and Figure S2 give complete predictive and calibration results.

\subsection{Real data example 2: Voting structure in the 117th U.S. House}
\label{sec:voteview-experiment}

The analysis uses the official Voteview record for the 117th U.S. House \citep{lewis_2026_VoteviewCongressionalRollCall}. For member \(i\) and annotated roll call \(j\), define \(Y_{ij}=1\) for cast codes 1, 2, or 3 and \(Y_{ij}=0\) for codes 4, 5, or 6. Absences and present votes are outside \(\Omega\). Thus
\[
Y_{ij}\mid(i,j)\in\Omega\sim\Bern(P_{ij}),
\qquad
\logit(P_{ij})=\alpha_i+\beta_j+\lambda(z_i\mid w_j)_o,
\]
and \(P_{ij}\) is the probability of a yea among recorded votes. The analysis contains 49,591 votes from 432 members and 118 roll calls with Congressional Research Service policy areas and recorded bill identifiers.

Five seeded masks independently assign 80\% of \(\Omega\) to training while retaining every active member and roll call. Each roll call has anchor path \(h_j=(\text{policy area}_j,\text{bill type}_j,\text{roll call}_j)\). Party labels are withheld during fitting. The goal is to measure predictive structure beyond member and roll-call propensities and to assess whether policy anchors yield an interpretable voting geometry. HELPI variants and all competing latent models use five starts.

\begin{figure}[tbp]
\centering
\includegraphics[width=\textwidth]{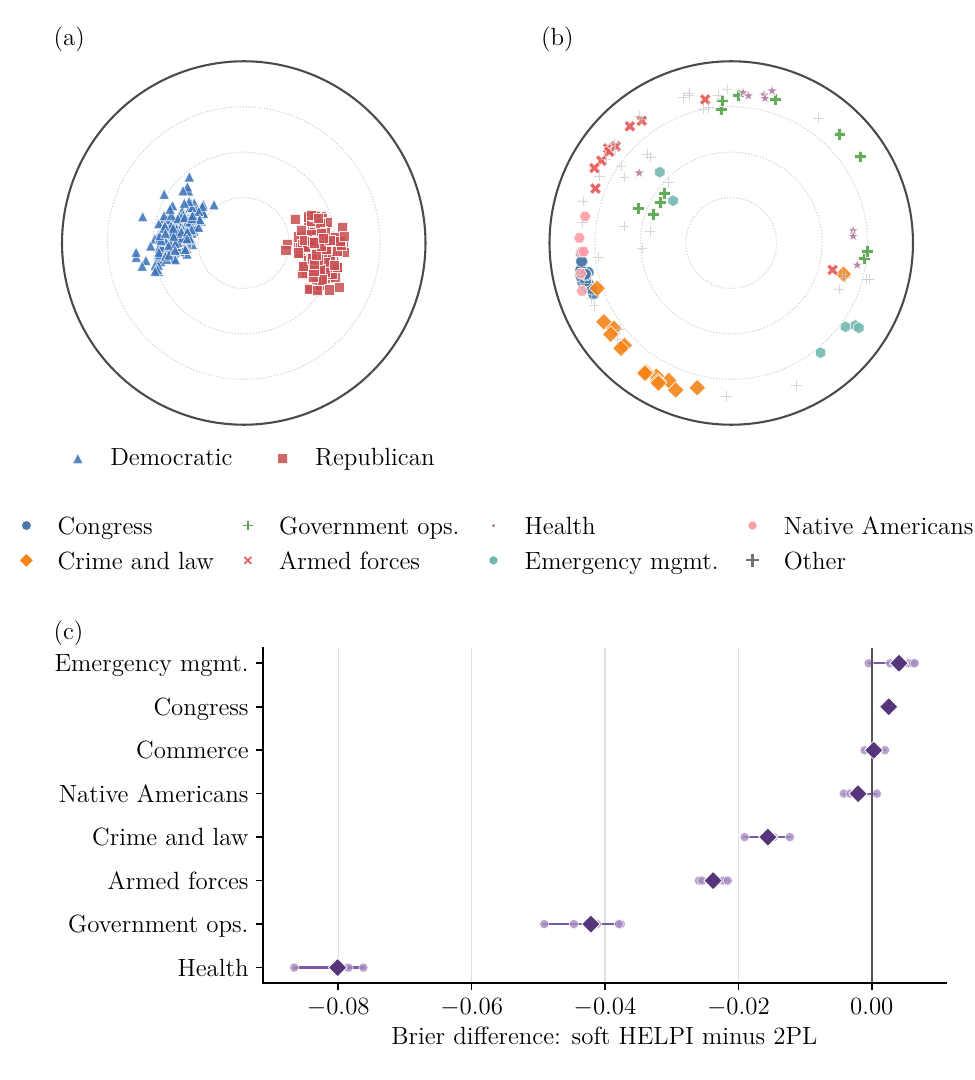}
\caption{Soft-anchor HELPI analysis of 117th House roll calls. Panel (a) places members in a Poincar\'e disk and shows two separated party clusters. Triangles mark Democratic members, and squares mark Republican members. Party labels were not used in fitting. Panel (b) distributes roll calls around a second disk and distinguishes the seven largest policy areas through marker shape and fill. Panel (c) gives soft HELPI minus 2PL Brier scores for the eight largest areas. Dots are mask results, diamonds are mask means, and negative values favor HELPI.}
\label{fig:voteview-geometry}
\end{figure}

Soft HELPI averages \(0.9902\) in AUC and \(0.0308\) in Brier score. Its mean AUC gains over main effects and 2PL are \(0.0809\) and \(0.0194\), with Brier reductions of \(0.0800\) and \(0.0236\). Relative to the AME and Euclidean fits, unanchored HELPI differs by at most \(0.0010\) in AUC and \(0.0013\) in Brier score. Policy-specific improvements are largest for Health and Government Operations and Politics, with further gains for Armed Forces and Crime and Law Enforcement.

Figure~\ref{fig:voteview-geometry} shows that the fitted member geometry separates the two major parties without using party labels. The roll-call geometry also retains visible policy organization. Soft anchoring combines the external taxonomy with voting coalitions learned from outcomes, and its gains over additive benchmarks persist across masks. Fixed anchors improve on main effects, while flexible variants accommodate structure beyond coarse policy labels. These embeddings summarize predictive voting patterns and are not causal estimates of policy effects. Supplementary Section S6.7 reports fixed-anchor, mask-specific, and calibration results.

A Markov chain Monte Carlo (MCMC) audit evaluates one simulation array and connected Junyi and Voteview subsets. All 18 invariant diagnostics satisfy the prespecified split-\(\widehat R\) and effective-sample-size thresholds \citep{vehtari_2021_RankNormalizationFoldingLocalization}. Supplementary Figure S5 and Table S14 give target definitions, chain lengths, and complete diagnostics. These subset calculations assess sampler implementation, while full-data conclusions rely on held-out prediction, multiple starts, and calibration.
\section{Discussion}
\label{sec:discussion}

HELPI separates the chosen hyperbolic representation from the geometric quantity that the bipartite likelihood can identify. With unrestricted main effects, radial depth is additive, while $\mathcal H=\lambda\projA\bG$ captures the estimable hierarchy signal. The identification results characterize equivalent predictors, collapse, curvature confounding, and local branch resolvability. In simulations, HELPI improves on main effects under hierarchical generation and captures nonadditive dependence under complementary structures. The Junyi analysis adds a taxonomy-based geometric summary while remaining close to item-response benchmarks. In Voteview, flexible HELPI variants improve on additive and 2PL benchmarks and recover party and policy organization without using party labels in fitting.

The present theory has a deliberately defined scope. Posterior contraction assumes fixed latent dimension, bounded depth, and a well-specified binary response model. These conditions provide a tractable benchmark for the identified hierarchy signal, not requirements for every useful application. Observed-design contraction addresses recorded dyads and does not imply full-array recovery under informative exposure. Likewise, the exact variational result concerns the global mean-field optimum, while practical fits are assessed through multiple starts, invariant diagnostics, boundary checks, and held-out calibration.

These boundaries suggest concrete extensions. Joint exposure-response models could separate opportunity from outcome, and growing dimension or adaptive depth could accommodate more complex hierarchies. Further work on uncertainty for contrasts of $\mathcal H$ and local variational solutions would connect theory more closely to computation. External taxonomies could be treated as uncertain inputs, while nonbinary outcomes require response-specific algorithms and asymptotic analysis.


\clearpage
\phantomsection
\label{supplementary-material}

\setcounter{section}{0}
\setcounter{figure}{0}
\setcounter{table}{0}
\setcounter{equation}{0}
\renewcommand{\thesection}{S\arabic{section}}
\renewcommand{\thefigure}{S\arabic{figure}}
\renewcommand{\thetable}{S\arabic{table}}
\renewcommand{\theequation}{S\arabic{equation}}
\renewcommand{\theHsection}{supp.section.\arabic{section}}
\renewcommand{\theHfigure}{supp.figure.\arabic{figure}}
\renewcommand{\theHtable}{supp.table.\arabic{table}}
\renewcommand{\theHequation}{supp.equation.\arabic{equation}}

\begin{center}
{\LARGE\bf Supplementary Material for}\\[0.6em]
{\Large\bf \textquotedblleft Hyperbolic Latent Position Models for Hierarchical Bipartite Data\textquotedblright}
\end{center}
\bigskip

This supplement contains all proofs for the results stated in the main manuscript. We also record the complete empirical designs, numerical summaries, calibration results, and convergence audit. Throughout, HELPI denotes the HypErbolic Latent Position model for bipartite Interaction data introduced in the main manuscript. Restatements and proof headings for main-text results retain the corresponding main-manuscript numbers; independently numbered supplementary results, sections, and displays carry an ``S'' prefix.

\section{Identification proofs}

\subsection{Equivalence, invariance, and collapse}

\begin{mainthmequivalence}[Distance-Gromov-product equivalence]
	For any root $o$,
	\[
	\alpha_i+\beta_j-\gamma d(z_i,w_j)
	=
	\{\alpha_i-\gamma d(o,z_i)\}
	+
	\{\beta_j-\gamma d(o,w_j)\}
	+
	2\gamma(z_i\mid w_j)_o.
	\]
	Conversely, every shared-ancestry predictor with coefficient $\lambda>0$ has a distance representation with $\gamma=\lambda/2$.
\end{mainthmequivalence}

\begin{proof}[Proof of Theorem~4.1]
	The definition of the rooted Gromov product gives
	\[
	2\gamma(z_i\mid w_j)_o
	=
	\gamma d(o,z_i)+\gamma d(o,w_j)-\gamma d(z_i,w_j).
	\]
	Substitution proves the forward identity. For the converse, set $\gamma=\lambda/2$ and add $\gamma d(o,z_i)$ and $\gamma d(o,w_j)$ to the two main effects.
\end{proof}

\begin{mainthmpopulation}[Population equivalence modulo additive effects]
	For hierarchical pairs $(\lambda,\bG)$ and $(\lambda',\bG')$, there exist $A,A'\in\calA$ satisfying
	\[
	A+\lambda\bG=A'+\lambda'\bG'
	\]
	if and only if $\lambda'\bG'-\lambda\bG\in\calA$. Equivalently,
	\[
	\lambda\projA\bG=\lambda'\projA\bG'.
	\]
	The same identified hierarchy signal is $-\gamma\projA\bD$ in the distance parameterization and is invariant to the root.
\end{mainthmpopulation}

\begin{proof}[Proof of Theorem~4.2]
	If the displayed predictors are equal, then
	\[
	\lambda'\bG'-\lambda\bG=A-A'\in\calA.
	\]
	Conversely, if $C=\lambda'\bG'-\lambda\bG\in\calA$, take $A=C$ and $A'=0$. The projection statement follows from $X\in\calA$ if and only if $\projA X=0$.
	
	The distance-Gromov identity from Theorem~4.1 gives
	\[
	-\gamma\projA\bD=2\gamma\projA\bG=\lambda\projA\bG.
	\]
	If $o'$ is another root, then
	\[
	G_o(z_i,w_j)-G_{o'}(z_i,w_j)
	=
	\frac12\{d(o,z_i)-d(o',z_i)\}
	+
	\frac12\{d(o,w_j)-d(o',w_j)\},
	\]
	which is additive in $i$ and $j$. Projection removes this difference.
\end{proof}

\begin{mainpropcurvature}[Curvature-coefficient confounding in the likelihood]
	Let $\mathfrak g_{\kappa_c}=\kappa_c^{-1}\mathfrak g_1$. Under the identity correspondence of points,
	\[
	d_{\kappa_c}=\kappa_c^{-1/2}d_1,
	\qquad
	(\cdot\mid\cdot)_{o,\kappa_c}=\kappa_c^{-1/2}(\cdot\mid\cdot)_{o,1}.
	\]
	Hence the likelihood depends on curvature and the geometry coefficient only through $\gamma/\sqrt{\kappa_c}$ or $\lambda/\sqrt{\kappa_c}$.
\end{mainpropcurvature}

\begin{proof}[Proof of Proposition~4.1]
	Multiplying a Riemannian metric by $\kappa_c^{-1}$ multiplies all path lengths by $\kappa_c^{-1/2}$ and changes sectional curvature from $-1$ to $-\kappa_c$. The Gromov product is a linear combination of three distances and scales identically.
\end{proof}

\begin{mainpropcollapse}[Ordered single-branch collapse]
	If all positions lie on one geodesic ray and one node type is uniformly deeper than the other, then $\bD\in\calA$ and $\projA\bG=0$.
\end{mainpropcollapse}

\begin{proof}[Proof of Proposition~4.2]
	Suppose $\ell_i^{(z)}>\ell_j^{(w)}$ for all $i,j$. Then $D_{ij}=\ell_i^{(z)}-\ell_j^{(w)}\in\calA$ and
	\[
	G_{ij}=\frac12\{\ell_i^{(z)}+\ell_j^{(w)}-(\ell_i^{(z)}-\ell_j^{(w)})\}=\ell_j^{(w)},
	\]
	which is column additive. The other ordering is symmetric.
\end{proof}

\begin{mainproptree}[Tree interpretation]
	In an exact rooted tree metric, $G_{ij}$ equals the depth of the lowest common ancestor of $z_i$ and $w_j$.
\end{mainproptree}

\begin{proof}[Proof of Proposition~4.3]
	If $c$ is the lowest common ancestor, then
	\[
	d(z_i,w_j)=d(o,z_i)+d(o,w_j)-2d(o,c).
	\]
	Substitution into the Gromov product gives $G_{ij}=d(o,c)$.
\end{proof}

\subsection{Constructive local image and branch resolvability}

Let $\calX^\circ$ be the smooth position domain defined in the main manuscript and let
\[
\Psi(\tau_{1:n},\upsilon_{1:m})=
\left[G_o\{\Exp_o(\tau_i),\Exp_o(\upsilon_j)\}\right]_{ij}.
\]

\begin{lemma}[Generic maximal-rank local image]
	\label{supp:localimage}
	On every connected component $\calC$ of $\calX^\circ$, the map $\Psi$ is real analytic. If
	\[
	d_\star=\max_{x\in\calC}\operatorname{rank}D\Psi_x,
	\]
	then the maximal-rank locus $\calC_{\max}=\{x:\operatorname{rank}D\Psi_x=d_\star\}$ is open and dense in $\calC$, and its complement is contained in a proper real-analytic subset of Lebesgue measure zero. For every $x_0\in\calC_{\max}$, there is a neighborhood $\mathcal U$ on which the rank is constant and
	\[
	\mathfrak S_G=\Psi(\mathcal U)
	\]
	is an embedded local submanifold with tangent space
	\[
	T_{\bG}\mathfrak S_G=\operatorname{ran}D\Psi_{x_0}.
	\]
	The tangent directions generated by simultaneous $O(p)$ rotations fixing the root belong to $\ker D\Psi_{x_0}$.
\end{lemma}

\begin{proof}
	Away from the root and cross-position coincidences, the norm, exponential map, Lorentz inner product, and $\arcosh$ are real analytic. Hence $\Psi$ is real analytic. Rank is lower semicontinuous, so the maximal-rank locus is open. Fix $x_0$ at which rank $d_\star$ is attained and choose a maximal minor that is nonzero at $x_0$. By the identity theorem, this minor is not identically zero on the connected component containing $x_0$. Its zero set is a proper real-analytic subset with empty interior and Lebesgue measure zero \citep{krantz_2002_PrimerRealAnalytic,mityagin_2020_ZeroSetReal}. The nonvanishing loci of all maximal minors therefore form an open dense set on that component. The excluded root and coincidence sets may divide the domain into components with different collision-free branch arrangements. On the maximal-rank locus, the constant-rank theorem yields the stated local image and tangent space.
	
	A simultaneous rotation $Q\in O(p)$ fixing $o$ preserves $d(o,z_i)$, $d(o,w_j)$, and $d(z_i,w_j)$, hence preserves every Gromov product. Differentiating the orbit shows that its tangent directions lie in the kernel of $D\Psi_{x_0}$.
\end{proof}

\begin{mainthmlocalid}[Local branch identification]
	Let $x_0\in\calC_{\max}$, let $\bG=\Psi(x_0)$, and let $E_{\calA}$ and $E_G$ be orthonormal bases of $\calA$ and $\operatorname{ran}D\Psi_{x_0}$. If
	\[
	J_{\rm br}=[E_{\calA},\operatorname{vec}(\bG),\lambda E_G]
	\]
	has full column rank, then the map $(A,\lambda,\bG)\mapsto A+\lambda\bG$ is one-to-one in a neighborhood of the reference point on $\calA\times(0,\infty)\times\mathfrak S_G$. Full column rank requires
	\[
	(n+m-1)+1+\operatorname{rank}D\Psi_{x_0}\le nm.
	\]
	When the residual stabilizer is discrete, the root-fixing $O(p)$ orbit has dimension $p(p-1)/2$ and lies in $\ker D\Psi_{x_0}$, so
	\[
	\operatorname{rank}D\Psi_{x_0}\le (n+m)p-\frac{p(p-1)}2.
	\]
	More generally, the actual orbit dimension is subtracted.
\end{mainthmlocalid}

\begin{proof}[Proof of Theorem~4.3]
	The differential is
	\[
	DF(\dot A,c,\dot\bG)=\dot A+c\bG+\lambda\dot\bG.
	\]
	In the selected bases, its matrix is $J_{\rm br}$. Full column rank is therefore equivalent to injectivity. The immersion theorem gives local injectivity in a neighborhood of the reference point. The dimension inequality is necessary because the number of columns of $J_{\rm br}$ cannot exceed $nm$.
\end{proof}

\section{Geometry and link lemmas}

\begin{lemma}[Compact root coordinates and Gromov-product Lipschitzness]
	\label{supp:geometrylip}
	For $\tau,\tau'\in B_{\mathrm E}(R_0)$,
	\[
	\|\tau-\tau'\|
	\le d\{\Exp_o(\tau),\Exp_o(\tau')\}
	\le \frac{\sinh R_0}{R_0}\|\tau-\tau'\|,
	\]
	where the ratio is interpreted as one at $R_0=0$. Moreover,
	\[
	|G_o(z,w)-G_o(z',w')|
	\le d(z,z')+d(w,w'),
	\qquad 0\le G_o(z,w)\le R_0
	\]
	for positions in the hyperbolic ball of radius $R_0$.
\end{lemma}

\begin{proof}
	Hyperbolic space is a Hadamard manifold, so the exponential map is distance nondecreasing by the Cartan-Hadamard/Rauch comparison theorem, giving the lower bound with constant one. The singular values of $D\Exp_o$ are one in the radial direction and $\sinh \ell/\ell$ in tangential directions. The straight segment joining $\tau$ and $\tau'$ stays inside the Euclidean ball, so integrating the derivative bound along that segment gives the upper constant $\sinh R_0/R_0$. Equivalently, one may use geodesic convexity of hyperbolic balls and the corresponding bound for $D\Log_o$ \citep{bridson_1999_MetricSpacesNonpositive}.
	
	For the Gromov product, apply the reverse triangle inequality to the two root distances and the ordinary triangle inequality to the cross-distance. The range follows from the triangle inequality and from $d(o,z),d(o,w)\le R_0$.
\end{proof}

For Bernoulli probabilities $\pi,\pi'$, define
\[
\hBer^2(\pi,\pi')=(\sqrt\pi-\sqrt{\pi'})^2+
(\sqrt{1-\pi}-\sqrt{1-\pi'})^2.
\]
For a single dyad, let $f_\eta$ be the Bernoulli mass function with success probability $H(\eta)$, and define
\[
\ell_{\eta_0,\eta}(Y)=\log\frac{f_{\eta_0}(Y)}{f_\eta(Y)},
\qquad
\KL_{\eta_0,\eta}=\E_{\eta_0}\ell_{\eta_0,\eta}(Y),
\]
\[
V_{2,0}(\eta_0,\eta)=
\E_{\eta_0}\{\ell_{\eta_0,\eta}(Y)-\KL_{\eta_0,\eta}\}^2,
\qquad
\Renyi(f\|g)=\log\int \frac{f^2}{g}.
\]

\begin{lemma}[Global Hellinger regularity and local divergence regularity]
	\label{supp:link}
	For both the logistic and probit links, there is a universal $C_H<\infty$ such that
	\[
	\hBer^2\{H(\eta),H(\eta')\}\le C_H(\eta-\eta')^2
	\]
	for all $\eta,\eta'\in\R$. If $|\eta_0|\le M_0$, then for all $\eta\in\R$,
	\[
	\KL_{\eta_0,\eta}\le C(M_0)(\eta-\eta_0)^2.
	\]
	If additionally $|\eta-\eta_0|\le1$, then
	\[
	\KL_{\eta_0,\eta}+V_{2,0}(\eta_0,\eta)+\Renyi(f_{\eta_0}\|f_\eta)
	\le C(M_0)(\eta-\eta_0)^2.
	\]
\end{lemma}

\begin{proof}
	For the logistic link, the derivatives of $\sqrt\Lambda$ and $\sqrt{1-\Lambda}$ are globally bounded. For probit, the derivatives are
	\[
	\frac{\phi_0(\eta)}{2\sqrt{\Phi_0(\eta)}}
	\quad\text{and}\quad
	\frac{\phi_0(\eta)}{2\sqrt{1-\Phi_0(\eta)}},
	\]
	which are continuous and tend to zero in both tails by Mills-ratio asymptotics. The global Hellinger bound follows from the mean-value theorem.
	
	For a fixed bounded truth, the Bernoulli Kullback-Leibler divergence divided by $(\eta-\eta_0)^2$ extends continuously at $\eta=\eta_0$ and remains bounded in both tails. For logistic this also follows directly from the Bregman representation with second derivative at most $1/4$; for probit it follows from the quadratic tail growth of $-\log\Phi_0$ and $-\log(1-\Phi_0)$. The local second-moment and R\'enyi bounds follow by a second-order expansion on a compact neighborhood. The centered second moment is bounded by the uncentered quantity $\E_{\eta_0}\ell_{\eta_0,\eta}^2$, so the same quadratic estimate applies.
\end{proof}

\begin{lemma}[Predictor Lipschitz bound]
	\label{supp:predictor}
	Let $\theta=(\alpha,\beta,\lambda,\tau_{1:n},\upsilon_{1:m})$ and $\theta'$ be another parameter value, with all tangent coordinates in $B_{\mathrm E}(R_0)$ and $\lambda,\lambda'\ge0$. Then
	\begin{align*}
		|\eta_{ij}(\theta)-\eta_{ij}(\theta')|
		&\le |\alpha_i-\alpha_i'|+|\beta_j-\beta_j'|+R_0|\lambda-\lambda'|\\
		&\quad+C_{R_0}\max(\lambda,\lambda')
		\{\|\tau_i-\tau_i'\|+\|\upsilon_j-\upsilon_j'\|\},
	\end{align*}
	where $C_{R_0}=\sinh R_0/R_0$.
\end{lemma}

\begin{proof}
	Add and subtract $\lambda G_o\{\Exp_o(\tau_i'),\Exp_o(\upsilon_j')\}$. The coefficient-difference term is bounded by $R_0|\lambda-\lambda'|$. Lemma~\ref{supp:geometrylip} bounds the Gromov-product difference by $C_{R_0}$ times the two tangent-coordinate differences. Reversing the roles of the two parameter values yields the symmetric maximum.
\end{proof}

\section{Posterior contraction}

Let
\[
N=nm,
\qquad
d_N=(n+m)(p+1)+1,
\qquad
\eps_N^2=K_0\frac{d_N\log N}{N},
\]
and let $R_N=C_0\sqrt{d_N\log N}$. For $R>0$, define the scalar sieve
\[
\Theta_N(R)=\left\{\max_i|\alpha_i|\le R,\ \max_j|\beta_j|\le R,\ 0\le\lambda\le R,\ \|\tau_i\|,\|\upsilon_j\|\le R_0\right\}.
\]

\begin{lemma}[Sieve entropy]
	\label{supp:entropy}
	For sufficiently small $t>0$ and $R\ge1$,
	\[
	\log\cover\{t,\{P_\theta:\theta\in\Theta_N(R)\},h_N\}
	\le C d_N\log(CR/t).
	\]
\end{lemma}

\begin{proof}
	Use mesh $ct$ for every scalar coordinate and mesh $ct/R$ for each tangent coordinate. The side condition $R\ge1$ holds in every later application because both $R_N$ and $R_N(t)$ diverge. The scalar intervals require at most $CR/t$ points each; the tangent ball requires at most $(CR/t)^p$ points per factor. The logarithm of the product-net size is therefore bounded by $Cd_N\log(CR/t)$. On the sieve, Lemma~\ref{supp:predictor} has $\max(\lambda,\lambda')\le R$, so the selected meshes give a predictor error of order $t$. Lemma~\ref{supp:link} transfers the predictor net to an $h_N$-net with a universal constant.
\end{proof}

\begin{lemma}[Prior thickness with second moments and R\'enyi divergence]
	\label{supp:prior}
	For every $r_N\downarrow0$, there is a product neighborhood $\mathcal U_N$ of the truth such that
	\[
	\sup_{\theta\in\mathcal U_N}
	\frac1N\sum_{i,j}
	\left[
	\KL_{ij}(\theta_0,\theta)+V_{2,0;ij}(\theta_0,\theta)
	+\Renyi(P_{0,ij}\|P_{\theta,ij})
	\right]
	\le Cr_N^2
	\]
	and
	\[
	\Pi(\mathcal U_N)\ge\exp\{-Cd_N\log(1/r_N)\}.
	\]
	At $\lambda_0=0$, the $\lambda$-neighborhood is the one-sided interval $[0,r_N]$.
\end{lemma}

\begin{proof}
	Take coordinatewise scalar neighborhoods of radius $r_N$ and tangent-coordinate balls of radius $r_N$. The true coefficient and positions are bounded, so Lemma~\ref{supp:predictor} gives $\sup_{ij}|\eta_{ij}-\eta_{0,ij}|\le Cr_N$. The true predictors are uniformly bounded, and the neighborhood remains in a fixed larger compact interval. Lemma~\ref{supp:link} gives all three dyadwise bounds. Prior positivity gives mass at least $cr_N$ for each scalar coordinate, including the one-sided interval at zero, and at least $cr_N^p$ for each tangent factor. Multiplication gives the mass bound.
\end{proof}

\begin{lemma}[Sieve complement]
	\label{supp:sieve}
	For $C_0$ sufficiently large,
	\[
	\Pi\{\Theta_N(R_N)^c\}
	\le\exp\{-C_s d_N\log N\},
	\]
	where $C_s$ can be made arbitrarily large by increasing $C_0$.
\end{lemma}

\begin{proof}
	Each Gaussian or half-normal scalar tail is bounded by $C\exp(-cR_N^2)$. There are at most $n+m+1\le d_N$ scalar coordinates. Hence
	\[
	\Pi\{\Theta_N(R_N)^c\}
	\le Cd_N\exp\{-cC_0^2d_N\log N\}.
	\]
	The polynomial prefactor is absorbed into the exponential for large $N$. The tangent priors are supported on $B_{\mathrm E}(R_0)$.
\end{proof}

\begin{lemma}[Tests for Bernoulli product measures]
	\label{supp:tests}
	Let $P_0^{(N)}$ and $P_1^{(N)}$ be products of independent Bernoulli laws. For every $r>0$ and every $P_1^{(N)}$ with $h_N(P_1,P_0)>r$, there exists a test with type-I and uniform type-II errors over the $h_N$-ball of radius $r/2$ bounded by $\exp(-cNr^2)$ for a universal $c>0$.
\end{lemma}

\begin{proof}
	This is Lemma~2 of \citet{ghosal_2007_ConvergenceRatesPosterior}, applied to the average Hellinger semimetric for independent non-identically distributed observations. The lemma may give different constants for the two error probabilities. The constant $c$ in the statement is their minimum.
\end{proof}

\begin{mainthmcontraction}[Dense-array posterior contraction]
	For every $a_N\to\infty$,
	\[
	\Pi\{h_N(P_\theta,P_0)>a_N\eps_N\mid Y\}\to0
	\]
	in $P_0$-probability.
\end{mainthmcontraction}

\begin{proof}[Proof of Theorem~6.1]
	Following the foundational posterior-rate framework of \citet{ghosal_2000_ConvergenceRatesPosterior}, we verify Theorem 4 of \citet{ghosal_2007_ConvergenceRatesPosterior}, with $k=2$ and semimetric $d_n=h_N$. Lemmas~\ref{supp:entropy} and \ref{supp:tests} give the local testing condition after the shell-separation constant is chosen so that the testing exponent dominates the entropy constant. Lemma~\ref{supp:prior} supplies the Kullback-Leibler and centered $V_{2,0}$ neighborhood required by condition (3.5); write its prior mass lower bound as $\exp(-C_{\rm pr}N\eps_N^2)$. Lemma~\ref{supp:sieve} gives
	\[
	\Pi\{\Theta_N(R_N)^c\}
	\le \exp\{-(C_s/K_0)N\eps_N^2\}.
	\]
	First choose $K_0$ to dominate the entropy constant. Then choose $C_0$, hence $C_s$, so that $C_s/K_0>C_{\rm pr}+4$. This verifies the sieve-to-prior-mass ratio condition (3.3). The annular prior-ratio condition (3.4) follows because its numerator is at most one and the denominator is at least $\exp(-C_{\rm pr}N\eps_N^2)$. The conclusion of Theorem 4 yields contraction for every diverging multiplier $a_N$.
\end{proof}

\begin{mainpropobserved}[Conditional contraction for an observed dyad set]
	Let $\Omega_N$ be deterministic with $M_N=|\Omega_N|$. After deleting rows and columns with no observed dyads, let $n_{\Omega,N}$ and $m_{\Omega,N}$ be the active-node counts and set
	\[
	d_{\Omega,N}=(n_{\Omega,N}+m_{\Omega,N})(p+1)+1.
	\]
	Define
	\[
	h_{\Omega_N}^2(P,Q)=M_N^{-1}\sum_{(i,j)\in\Omega_N}\hBer^2(P_{ij},Q_{ij}).
	\]
	If $d_{\Omega,N}\log M_N=o(M_N)$ and
	\[
	\eps_{\Omega,N}^2=K_0\frac{d_{\Omega,N}\log M_N}{M_N},
	\]
	then the posterior based on $L_{\Omega_N}$ contracts at rate $\eps_{\Omega,N}$ in $h_{\Omega_N}$.
\end{mainpropobserved}

\begin{proof}[Proof of Proposition~6.1]
	Delete all isolated rows and columns first; their parameters do not enter the likelihood or the observed-dyad semimetric. Apply Lemmas~\ref{supp:entropy} through \ref{supp:tests} to the resulting active parameterization. Restrict every product and semimetric sum to $\Omega_N$, and replace $(N,d_N,\log N)$ by $(M_N,d_{\Omega,N},\log M_N)$. In particular, use the scalar sieve radius
	\[
	R_{\Omega,N}=C_0\sqrt{d_{\Omega,N}\log M_N}.
	\]
	The entropy bound is $O\{d_{\Omega,N}\log M_N\}$, the prior-thickness exponent has the same order, and likelihood divergences and test exponents scale with $M_N$. Theorem~4 of \citet{ghosal_2007_ConvergenceRatesPosterior} therefore gives the result. Replacing $\log M_N$ by the larger $\log N$ gives a valid but looser bound.
\end{proof}

\begin{maincorsignal}[Contraction of the identified hierarchy signal]
	Let
	\[
	g_{B_{\rm clip}}(t)=\clip\{H^{-1}(t),[-B_{\rm clip},B_{\rm clip}]\},
	\qquad B_{\rm clip}\ge\|\bM_0\|_\infty,
	\]
	and set $\mathcal H_{B_{\rm clip}}(P)=\projA[g_{B_{\rm clip}}(P)]$. Then for every $a_N\to\infty$,
	\[
	\Pi\left\{N^{-1/2}\|\mathcal H_{B_{\rm clip}}(P_\theta)-\mathcal H_{B_{\rm clip}}(P_0)\|_F
	>a_N C_{B_{\rm clip}}\eps_N\mid Y\right\}\to0.
	\]
\end{maincorsignal}

\begin{proof}[Proof of Corollary~6.1]
	The clipped inverse link is globally Lipschitz with constant $C_{B_{\rm clip}}$. Orthogonal projection is nonexpansive, and $|\pi-\pi'|\le\hBer(\pi,\pi')$. Hence
	\[
	N^{-1/2}\|\mathcal H_{B_{\rm clip}}(P)-\mathcal H_{B_{\rm clip}}(Q)\|_F
	\le C_{B_{\rm clip}}h_N(P,Q).
	\]
	Apply Theorem~6.1.
\end{proof}

\section{Exact mean-field variational risk}

For each $N$, let $v_N=\eps_N^4$ and $R_N^{\rm mf}=C_{\rm mf}\sqrt{d_N\log N}$. Let $\calQ_N^{\rm mf}$ be the product family described in the main manuscript. Scalar means are bounded by $R_N^{\rm mf}$, and scalar variances and tangent covariance eigenvalues lie in $[v_N,v_N^{-1}]$. Tangent means lie in $B_{\mathrm E}(R_0)$, and tangent factors are Gaussian laws truncated to that ball. Every $Q\in\calQ_N^{\rm mf}$ is absolutely continuous with respect to the bounded-depth prior. The variational parameter set is compact for fixed $N$.

\begin{lemma}[Existence of the exact variational optimum]
	\label{supp:vbexist}
	The map $Q\mapsto\KL\{Q\|\Pi(\cdot\mid Y)\}$ is finite and lower semicontinuous on the compact variational parameter set. Consequently a global minimizer exists.
\end{lemma}

\begin{proof}
	The parameter-to-distribution map is weakly continuous. Gaussian laws vary weakly continuously with their means and nondegenerate covariance matrices. Truncated Gaussian laws on the fixed ball have densities and normalizing constants that depend continuously on the parameters. Relative entropy is jointly lower semicontinuous under weak convergence. The posterior measure is fixed, so composition gives lower semicontinuity of the objective.
	
	The objective is finite on the family. The scalar factors have finite second moments, the scalar prior negative log densities grow at most quadratically, and the tangent prior density is bounded below on the compact ball. The Bernoulli negative log likelihood grows at most linearly in the predictor for logit and at most quadratically for probit, while $G_o$ is bounded by $R_0$. Finally, the entropy of every factor is finite because covariance eigenvalues are bounded away from zero and infinity for fixed $N$. Compactness and lower semicontinuity yield a minimizer.
\end{proof}

\begin{lemma}[Mean-field comparison distribution]
	\label{supp:comparison}
	There exists $Q_{0,N}\in\calQ_N^{\rm mf}$ such that
	\[
	\KL(Q_{0,N}\|\Pi_0)\le Cd_N\log(1/\eps_N)
	\]
	and
	\[
	\E_{Q_{0,N}}\KL(P_0^{(N)}\|P_\theta^{(N)})
	\le CN\eps_N^2.
	\]
\end{lemma}

\begin{proof}
	Use Gaussian scalar factors centered at the truth with variance $\eps_N^2$, a positive-truncated Gaussian for $\lambda$, and multivariate Gaussian tangent factors centered at the true coordinates with covariance $\eps_N^2I_p$, truncated to $B_{\mathrm E}(R_0)$. Because the true tangent coordinates lie a fixed distance inside the ball, the tangent truncation probability tends to one exponentially fast.
	
	The scalar Gaussian-to-prior KL divergence is bounded by $C+\log(1/\eps_N)$ per coordinate. The positive-truncated coefficient factor has the same order, including at $\lambda_0=0$. Each $p$-dimensional tangent factor contributes $p\log(1/\eps_N)+O(1)$ from its entropy, while the prior log density is bounded on the compact ball. Summing gives the first assertion.
	
	For the sampling KL, Assumption~6.1 of the main manuscript implies the uniform true-predictor bound
	\[
	M_0:=2B_0+\lambda_{\max,0}R_0,
	\qquad
	\max_{i,j}|\eta_{0,ij}|\le M_0.
	\]
	Apply Lemma~\ref{supp:predictor} with $\theta'=\theta_0$. Squaring the resulting five-term bound and taking expectations gives
	\begin{align*}
		\E_{Q_{0,N}}(\eta_{ij}-\eta_{0,ij})^2
		&\le C\Bigl[
		\E(\alpha_i-\alpha_{0i})^2+
		\E(\beta_j-\beta_{0j})^2+
		\E(\lambda-\lambda_0)^2\\
		&\qquad+
		\E\{\max(\lambda,\lambda_0)^2\|\tau_i-\tau_{0i}\|^2\}
		+
		\E\{\max(\lambda,\lambda_0)^2\|\upsilon_j-\upsilon_{0j}\|^2\}
		\Bigr].
	\end{align*}
	The coefficient factor is independent of the tangent factors under $Q_{0,N}$. Hence, for example,
	\[
	\E\{\max(\lambda,\lambda_0)^2\|\tau_i-\tau_{0i}\|^2\}
	\le 2\{\E(\lambda^2)+\lambda_0^2\}\E\|\tau_i-\tau_{0i}\|^2
	\le C\eps_N^2.
	\]
	All other terms have the same order, so the predictor second moment is uniformly $O(\eps_N^2)$. Lemma~\ref{supp:link} gives a dyadwise expected sampling KL of the same order, and summing over $N$ dyads proves the second assertion.
\end{proof}

\begin{mainthmvariational}[Predictive risk of the exact mean-field optimum]
	Let $\widehat Q_N$ be the exact reverse-KL minimizer over $\calQ_N^{\rm mf}$. Then
	\[
	\E_{P_0}\int h_N^2(P_\theta,P_0)\,\widehat Q_N(\dd\theta)
	\le C\eps_N^2.
	\]
	The same rate holds for the squared normalized error of $\mathcal H_{B_{\rm clip}}(P_\theta)$.
\end{mainthmvariational}

\begin{proof}[Proof of Theorem~6.2]
	We apply Theorem~2.2 of \citet{zhang_2020_ConvergenceRatesVariational}. Their theorem assumes conditions (C1) through (C4), uses a normalized rate $\varepsilon_n$ satisfying $n\varepsilon_n^2\ge1$, and concludes expected loss of order $n\varepsilon_n^2$. We identify their sample size $n$ with our $N$ and their normalized rate $\varepsilon_n$ with our $\eps_N$. With the loss
	\[
	L_N(P_\theta^{(N)},P_0^{(N)})=Nh_N^2(P_\theta,P_0),
	\]
	the unnormalized target scale is therefore $N\eps_N^2$. Their equation~(7) states the two comparison inequalities that together imply (C4).
	
	For every $t>\eps_N$, define a sieve by bounding the scalar coordinates at
	\[
	R_N(t)=C_0\sqrt N\,t
	\]
	and retaining the fixed tangent ball. Lemma~\ref{supp:entropy} gives
	\[
	\log\cover\{t/36,\Theta_N(R_N(t)),h_N\}
	\le Cd_N\log(C\sqrt N)
	\le C'd_N\log N
	\le (C'/K_0)Nt^2.
	\]
	Together with Lemma~\ref{supp:tests}, this verifies Zhang-Gao condition (C1) after the loss-separation constant is chosen and $K_0$ is sufficiently large. Sub-Gaussian tails give
	\[
	\Pi\{\Theta_N(R_N(t))^c\}\le\exp(-cC_0^2Nt^2),
	\]
	which is condition (C2) for sufficiently large $C_0$. Lemma~\ref{supp:prior} supplies a set of mass at least $\exp(-C_2N\eps_N^2)$ on which $\Renyi(P_0^{(N)}\|P_\theta^{(N)})\le C_3N\eps_N^2$, which is condition (C3).
	
	For condition (C4), Zhang and Gao define
	\[
	R(Q)=\frac1N\left\{\KL(Q\|\Pi_0)+\E_Q\KL(P_0^{(N)}\|P_\theta^{(N)})\right\}.
	\]
	Lemma~\ref{supp:comparison} gives $R(Q_{0,N})\le C\eps_N^2$, thereby verifying (C4); equivalently, its two terms are the two comparison inequalities in equation~(7) of \citet{zhang_2020_ConvergenceRatesVariational}. Theorem~2.2 therefore yields
	\[
	\E_{P_0}\int L_N(P_\theta^{(N)},P_0^{(N)})\,\widehat Q_N(\dd\theta)
	\le C N\eps_N^2.
	\]
	Dividing by $N$ proves the predictive risk bound. Finally, Corollary~6.1 gives pointwise
	\[
	N^{-1}\|\mathcal H_{B_{\rm clip}}(P)-\mathcal H_{B_{\rm clip}}(P_0)\|_F^2
	\le C_{B_{\rm clip}}^2h_N^2(P,P_0),
	\]
	and integration gives the identified-hierarchy-signal bound.
\end{proof}

\section{Computational details}
\label{supp:computation}

\subsection{Mean-field factors and predictor moments}

Let $z_i=\Exp_o(\tau_i)$ and $w_j=\Exp_o(\upsilon_j)$. The bounded-depth approximation uses the tangent-coordinate factors
\[
q_i^{(\tau)}(\tau_i)=\TN_p\{\mu_i^{(\tau)},\Sigma_i^{(\tau)};B_{\mathrm E}(R_0)\},
\qquad
q_j^{(\upsilon)}(\upsilon_j)=\TN_p\{\mu_j^{(\upsilon)},\Sigma_j^{(\upsilon)};B_{\mathrm E}(R_0)\}.
\]
The remaining factors are
\[
q_\alpha=\prod_i\Normal(\mu_{\alpha i},s_{\alpha i}^2),
\quad
q_\beta=\prod_j\Normal(\mu_{\beta j},s_{\beta j}^2),
\quad
q_\lambda=\TN_{[0,\infty)}(\mu_\lambda,s_\lambda^2).
\]
Write
\[
\bar G_{ij}=\E_q(G_{ij}),
\quad
m_{G,ij}^{(2)}=\E_q(G_{ij}^2),
\quad
\bar\lambda_q=\E_q(\lambda),
\quad
m_{\lambda,q}^{(2)}=\E_q(\lambda^2).
\]
For each observed dyad, define $\mu_{ij}^{\eta}=\E_q(\eta_{ij})$ and $\nu_{ij}^{\eta}=\E_q(\eta_{ij}^2)$. The quantity $\nu_{ij}^{\eta}$ is the raw second moment. Under the mean-field factorization, these moments are given by Equations~\eqref{supp:predictor-moments}:
\begin{subequations}\label{supp:predictor-moments}
	\begin{align}
		\mu_{ij}^{\eta}
		&=\mu_{\alpha i}+\mu_{\beta j}+\bar\lambda_q\bar G_{ij},\\
		\nu_{ij}^{\eta}
		&=s_{\alpha i}^2+s_{\beta j}^2+(\mu_{\alpha i}+\mu_{\beta j})^2
		+2(\mu_{\alpha i}+\mu_{\beta j})\bar\lambda_q\bar G_{ij}
		+m_{\lambda,q}^{(2)}m_{G,ij}^{(2)}.
	\end{align}
\end{subequations}

\subsection{Local and scalar coordinate updates}

For the logistic link, define $\kappa_{ij}=Y_{ij}-1/2$. P\'olya-Gamma augmentation gives
\begin{equation}\label{supp:pg-update}
	q(\omega_{ij})=\PG(1,\xi_{ij}),
	\qquad
	\xi_{ij}=\sqrt{\nu_{ij}^{\eta}},
	\qquad
	\bar\omega_{ij}=\frac{\tanh(\xi_{ij}/2)}{2\xi_{ij}},
	\qquad
	\widetilde y_{ij}=\frac{\kappa_{ij}}{\bar\omega_{ij}}.
\end{equation}
Equation~\eqref{supp:pg-update} is the P\'olya-Gamma form of the Jaakkola-Jordan variational update \citep{jaakkola_2000_BayesianParameterEstimation,polson_2013_BayesianInferenceLogistic,durante_2019_ConditionallyConjugateMeanField}.

For the probit link, introduce $W_{ij}^{\rm pr}=\eta_{ij}+e_{ij}^{\rm pr}$, where $e_{ij}^{\rm pr}\sim\Normal(0,1)$ and $Y_{ij}=1\{W_{ij}^{\rm pr}>0\}$. The optimal local factor is
\begin{equation}\label{supp:probit-update}
	q(w_{ij}^{\rm pr})=\TN_{\{(2Y_{ij}-1)w>0\}}(\mu_{ij}^{\eta},1),
	\qquad
	\widetilde y_{ij}=\mu_{ij}^{\eta}+(2Y_{ij}-1)
	\frac{\phi_0\{(2Y_{ij}-1)\mu_{ij}^{\eta}\}}
	{\Phi_0\{(2Y_{ij}-1)\mu_{ij}^{\eta}\}}.
\end{equation}
Equation~\eqref{supp:probit-update} uses the standard normal density $\phi_0$ and distribution function $\Phi_0$. For probit, set $\bar\omega_{ij}=1$.

Let $\Omega_i=\{j:(i,j)\in\Omega\}$ and $\Omega^j=\{i:(i,j)\in\Omega\}$. The scalar factors have exact coordinate updates
\begin{subequations}\label{supp:scalar-updates}
	\begin{align}
		s_{\alpha i}^{-2}
		&=\sigma_\alpha^{-2}+\sum_{j\in\Omega_i}\bar\omega_{ij},
		&
		\mu_{\alpha i}
		&=s_{\alpha i}^2\sum_{j\in\Omega_i}\bar\omega_{ij}
		\{\widetilde y_{ij}-\mu_{\beta j}-\bar\lambda_q\bar G_{ij}\},\\
		s_{\beta j}^{-2}
		&=\sigma_\beta^{-2}+\sum_{i\in\Omega^j}\bar\omega_{ij},
		&
		\mu_{\beta j}
		&=s_{\beta j}^2\sum_{i\in\Omega^j}\bar\omega_{ij}
		\{\widetilde y_{ij}-\mu_{\alpha i}-\bar\lambda_q\bar G_{ij}\},\\
		s_\lambda^{-2}
		&=\sigma_\lambda^{-2}+\sum_{(i,j)\in\Omega}\bar\omega_{ij}m_{G,ij}^{(2)},
		&
		\mu_\lambda
		&=s_\lambda^2\sum_{(i,j)\in\Omega}\bar\omega_{ij}\bar G_{ij}
		\{\widetilde y_{ij}-\mu_{\alpha i}-\mu_{\beta j}\}.
	\end{align}
\end{subequations}
After the scalar updates in Equations~\eqref{supp:scalar-updates}, the implementation recomputes $\bar\lambda_q$ and $m_{\lambda,q}^{(2)}$ from the positive-truncated factor rather than its untruncated Gaussian kernel.

\subsection{Geometry update and stopping rules}

The truncated tangent factors are nonconjugate. Let $\vartheta_t$ collect their mean and covariance parameters, and let $\mathcal F$ denote the feasible parameter set. We estimate the gradient of the augmented evidence lower bound by the score-function estimator $\widehat\nabla_\vartheta\mathcal L_t$, with control variates. One projected geometry step per outer sweep is
\begin{equation}\label{supp:tangent-update}
	\vartheta_{t+1}
	=\Pi_{\mathcal F}\!\left\{\vartheta_t+a_t
	\widehat\nabla_\vartheta\mathcal L_t\right\},
	\qquad
	a_t=a_0(t+t_0)^{-\zeta},
	\qquad \zeta\in(1/2,1].
\end{equation}
Here $\Pi_{\mathcal F}$ restricts tangent means to $B_{\mathrm E}(R_0-\delta_R)$ for a small numerical margin. It also restricts covariance eigenvalues to a prespecified interval $[v_{\min},v_{\max}]$. Pathwise gradients may replace score gradients when a differentiable truncated-Gaussian sampler is available.

The common-shift convention and invariant convergence summaries are updated after each sweep. Let $\mathcal L_t^{\rm avg}$ be the running-average or periodically refined objective, and let $\bP_t$ and $\mathcal H_t$ be the current fitted probability matrix and identified hierarchy signal. Then
\begin{subequations}\label{supp:centering-stopping}
	\begin{align}
		\bar\mu_{\alpha,t}
		&=n^{-1}\sum_{i=1}^n\mu_{\alpha i,t},
		&
		\mu_{\alpha i,t}&\leftarrow\mu_{\alpha i,t}-\bar\mu_{\alpha,t},
		&
		\mu_{\beta j,t}&\leftarrow\mu_{\beta j,t}+\bar\mu_{\alpha,t},\label{supp:centering-update}\\
		\mathcal C_t
		&=\left(
		\mathcal L_t^{\rm avg},
		\frac{\|\bP_t-\bP_{t-1}\|_F}{\sqrt{nm}},
		\frac{\|\mathcal H_t-\mathcal H_{t-1}\|_F}{\sqrt{nm}}
		\right).\label{supp:stopping-summary}
	\end{align}
\end{subequations}
The first line leaves every predictor unchanged. The second line monitors the objective together with two invariant fitted quantities.

The scalar updates are monotone for the augmented objective conditional on fixed geometry moments. Monte Carlo noise and the stochastic geometry step prevent exact coordinate-ascent monotonicity for the complete sweep. Stopping rules therefore use a running average or periodic high-accuracy evaluation. The reported solution is the independent restart with the largest refined objective estimate.

\section{Empirical design and complete results}
\label{supp:empirical}

We used the master seed \textsf{20260804} throughout this example. Each example was divided by task index into eight independent shards, and every worker used one numerical thread. All 3,697 scheduled tasks completed without an error record. 

The log and Brier scores are strictly proper scoring rules \citep{gneiting_2007_StrictlyProperScoring}. AUC, RMSE, and ECE denote area under the receiver operating characteristic curve, root mean squared error, and expected calibration error. AME and 2PL denote additive-and-multiplicative effects and two-parameter logistic.

\begin{table}[H]
	\centering
	\footnotesize
	\caption{Combined execution audit for the three experiment profiles.}
	\label{tab:task-audit}
	\begin{tabular}{lrrr}
		\toprule
		Experiment & Design units & Replicates, masks, or chains & Tasks \\
		\midrule
		Balanced hierarchy & 14 cells & 100 & 1,400 \\
		Collapse and weak branch & 7 spreads & 100 & 700 \\
		Complementary generators & 6 generator-size cells & 100 & 600 \\
		Analyst-choice sensitivity & 9 dimension-bound cells & 100 & 900 \\
		Junyi correctness & 4 windows by 8 models & 1 & 32 \\
		Voteview roll calls & 5 masks by 8 models & 1 & 40 \\
		Retailrocket conversion & 1 split by 8 models & 1 & 8 \\
		Junyi attempt process & 1 model & 1 & 1 \\
		Simulation and Junyi MCMC & 2 targets & 4 chains & 8 \\
		Voteview MCMC & 1 target & 8 chains & 8 \\
		\midrule
		Total & & & 3,697 \\
		\bottomrule
	\end{tabular}
	\TableNote{Design units are prespecified simulation cells, empirical splits, or MCMC targets. Tasks count independent model fits or chains; all 3,697 scheduled tasks completed.}
\end{table}
\subsection{Balanced hierarchy simulation}
\label{supp:balanced}

The full factorial contains 162 combinations of branching factor, depth, side length, hierarchy coefficient, and target event rate. A deterministic greedy design selected 14 cells that cover every factor level and every pair of factor levels. Table~\ref{tab:balanced-design} lists the selected cells. Each cell has 100 replicates and uses latent dimension four, depth bound five, and an 80\% training split.

\begin{table}[H]
	\centering
	\small
	\caption{Pairwise-covering cells in the balanced hierarchy experiment.}
	\label{tab:balanced-design}
	\begin{tabular}{rrrr@{\hspace{1.2em}}r@{\hspace{2.4em}}rrrr@{\hspace{1.2em}}r}
		\toprule
		$b$ & $h$ & Size & $\lambda$ & Rate & $b$ & $h$ & Size & $\lambda$ & Rate \\
		\midrule
		3 & 4 & 100 & 0.50 & 0.10 & 4 & 4 & 200 & 1.00 & 0.30 \\
		4 & 3 & 100 & 0.25 & 0.10 & 4 & 4 & 200 & 0.25 & 0.50 \\
		2 & 4 & 100 & 1.00 & 0.30 & 3 & 3 & 400 & 0.50 & 0.10 \\
		2 & 3 & 100 & 0.25 & 0.50 & 3 & 3 & 400 & 0.25 & 0.30 \\
		2 & 3 & 200 & 0.25 & 0.10 & 4 & 4 & 400 & 0.50 & 0.30 \\
		3 & 3 & 200 & 0.50 & 0.10 & 2 & 3 & 400 & 0.50 & 0.50 \\
		3 & 3 & 200 & 1.00 & 0.10 & 3 & 4 & 400 & 1.00 & 0.50 \\
		\bottomrule
	\end{tabular}
	\TableNote{$b$ is branching factor, $h$ is tree depth, Size is $n=m$, $\lambda$ is the hierarchy coefficient, and Rate is the target event probability. Each cell contains 100 replicates with $p=4$, $R_0=5$, and an 80\% training split.}
\end{table}

Table~\ref{tab:balanced-size} reports means after averaging equally over replicates in the cells assigned to each size. The size-specific compositions differ, so the rows do not estimate a controlled sample-size curve. Higher log score and AUC are better, while the remaining metrics are better when smaller.

\begingroup
\small
\setlength{\tabcolsep}{4pt}
\begin{longtable}{@{}r@{\hspace{1.2em}}lrrrrr@{}}
	\caption{Balanced hierarchy results by side length.}\label{tab:balanced-size}\\
	\toprule
	Size & Model & Log score & Brier & AUC & Probability RMSE & Signal RMSE \\
	\midrule
	\endfirsthead
	\toprule
	Size & Model & Log score & Brier & AUC & Probability RMSE & Signal RMSE \\
	\midrule
	\endhead
	100 & HELPI & -0.4708 & 0.1524 & 0.6376 & 0.1088 & 0.4592 \\
	100 & Main effects & -0.4788 & 0.1557 & 0.6198 & 0.1158 & 0.5071 \\
	100 & Euclidean & -0.4576 & 0.1465 & 0.6617 & 0.0920 & 0.3783 \\
	100 & Spherical & -0.4624 & 0.1482 & 0.6534 & 0.1010 & 0.4892 \\
	100 & AME & -0.4645 & 0.1495 & 0.6537 & 0.1020 & 0.4358 \\
	\addlinespace
	200 & HELPI & -0.4392 & 0.1396 & 0.6379 & 0.0849 & 0.4504 \\
	200 & Main effects & -0.4388 & 0.1395 & 0.6386 & 0.0840 & 0.4494 \\
	200 & Euclidean & -0.4352 & 0.1380 & 0.6592 & 0.0780 & 0.4055 \\
	200 & Spherical & -0.4381 & 0.1392 & 0.6526 & 0.0863 & 0.4809 \\
	200 & AME & -0.4387 & 0.1395 & 0.6408 & 0.0842 & 0.4486 \\
	\addlinespace
	400 & HELPI & -0.5487 & 0.1854 & 0.6673 & 0.0798 & 0.3754 \\
	400 & Main effects & -0.5573 & 0.1894 & 0.6517 & 0.0951 & 0.4558 \\
	400 & Euclidean & -0.5469 & 0.1845 & 0.6711 & 0.0751 & 0.3501 \\
	400 & Spherical & -0.5517 & 0.1863 & 0.6626 & 0.0869 & 0.4283 \\
	400 & AME & -0.5511 & 0.1863 & 0.6663 & 0.0835 & 0.3979 \\
	\bottomrule
\end{longtable}
\TableNote{Entries are replicate means, averaged equally over the selected design cells assigned to each side length. AUC is area under the receiver operating characteristic curve; RMSE is root mean squared error. Higher log score and AUC are better, while lower values are better for the remaining metrics.}
\endgroup

Table~\ref{tab:balanced-paired} gives paired HELPI-minus-baseline differences across all 1,400 datasets. Monte Carlo standard errors use the replicate differences. Positive differences favor HELPI for log score and AUC, while negative differences favor HELPI for the three error metrics.

\begingroup
\small
\setlength{\tabcolsep}{5pt}
\begin{longtable}{@{}llrrl@{}}
	\caption{Paired balanced-hierarchy contrasts.}\label{tab:balanced-paired}\\
	\toprule
	Metric & Baseline & Difference & MCSE & 95\% Monte Carlo interval \\
	\midrule
	\endfirsthead
	\toprule
	Metric & Baseline & Difference & MCSE & 95\% Monte Carlo interval \\
	\midrule
	\endhead
	Log score & Main effects & 0.00521 & 0.00037 & $[0.00449,0.00594]$ \\
	Log score & Euclidean & -0.00585 & 0.00042 & $[-0.00668,-0.00502]$ \\
	Log score & Spherical & -0.00176 & 0.00046 & $[-0.00267,-0.00085]$ \\
	Log score & AME & -0.00113 & 0.00028 & $[-0.00168,-0.00058]$ \\
	\addlinespace
	Brier & Main effects & -0.00231 & 0.00016 & $[-0.00263,-0.00200]$ \\
	Brier & Euclidean & 0.00258 & 0.00018 & $[0.00223,0.00292]$ \\
	Brier & Spherical & 0.00100 & 0.00019 & $[0.00063,0.00137]$ \\
	Brier & AME & 0.00056 & 0.00012 & $[0.00033,0.00080]$ \\
	\addlinespace
	AUC & Main effects & 0.01039 & 0.00076 & $[0.00889,0.01189]$ \\
	AUC & Euclidean & -0.01588 & 0.00090 & $[-0.01765,-0.01411]$ \\
	AUC & Spherical & -0.00809 & 0.00097 & $[-0.01000,-0.00618]$ \\
	AUC & AME & -0.00529 & 0.00061 & $[-0.00648,-0.00410]$ \\
	\addlinespace
	Probability RMSE & Main effects & -0.00713 & 0.00051 & $[-0.00813,-0.00613]$ \\
	Probability RMSE & Euclidean & 0.00897 & 0.00049 & $[0.00800,0.00993]$ \\
	Probability RMSE & Spherical & -0.00078 & 0.00060 & $[-0.00195,0.00038]$ \\
	Probability RMSE & AME & 0.00089 & 0.00032 & $[0.00026,0.00152]$ \\
	\addlinespace
	Signal RMSE & Main effects & -0.04208 & 0.00260 & $[-0.04717,-0.03699]$ \\
	Signal RMSE & Euclidean & 0.04820 & 0.00305 & $[0.04222,0.05419]$ \\
	Signal RMSE & Spherical & -0.03836 & 0.00399 & $[-0.04617,-0.03055]$ \\
	Signal RMSE & AME & -0.00070 & 0.00184 & $[-0.00430,0.00290]$ \\
	\bottomrule
\end{longtable}
\TableNote{Each difference is HELPI minus the named baseline over 1,400 paired datasets. MCSE denotes the Monte Carlo standard error of the paired mean. Positive values favor HELPI for log score and AUC; negative values favor HELPI for the error metrics.}
\endgroup

HELPI improves all five summaries relative to main effects. The Euclidean model has the best design-averaged prediction and signal recovery in these selected cells. AME has slightly better predictive averages and nearly identical signal RMSE, while the spherical model predicts more accurately but has larger signal RMSE than HELPI. Thus HELPI consistently captures nonadditive structure beyond main effects, although relative performance depends on the interaction geometry.

\subsection{Collapse and weak-branch results}
\label{supp:collapse-results}

The collapse experiment uses $80\times80$ arrays, hierarchy coefficient $0.75$, event rate $0.30$, and 100 replicates at each branch spread. Table~\ref{tab:collapse} reports the HELPI results. The population quotient singular value is zero at collapse, up to floating-point error, and increases with branch spread. The fitted quotient singular value remains nearly constant across the grid.

\begin{table}[H]
	\centering
	\scriptsize
	\caption{Collapse and weak-branch results. Values are replicate means.}
	\label{tab:collapse}
	\begin{tabular}{rrrrrrr}
		\toprule
		Spread & True $\sigma_{\min}$ & Fitted $\sigma_{\min}$ & Restart SD($\lambda$) & Log score & Probability RMSE & Signal RMSE \\
		\midrule
		0.00 & $6.63\times10^{-15}$ & 0.00419 & 0.01139 & -0.5934 & 0.0838 & 0.0299 \\
		0.01 & $1.10\times10^{-8}$ & 0.00432 & 0.01535 & -0.5883 & 0.0817 & 0.0300 \\
		0.03 & $1.00\times10^{-6}$ & 0.00438 & 0.01197 & -0.5911 & 0.0840 & 0.0300 \\
		0.10 & $9.45\times10^{-5}$ & 0.00407 & 0.01269 & -0.5888 & 0.0826 & 0.0350 \\
		0.30 & 0.003261 & 0.00433 & 0.01594 & -0.5941 & 0.0857 & 0.0925 \\
		0.60 & 0.008126 & 0.00433 & 0.01240 & -0.5921 & 0.0891 & 0.1661 \\
		1.00 & 0.014645 & 0.00426 & 0.01045 & -0.5921 & 0.0937 & 0.2156 \\
		\bottomrule
	\end{tabular}
	\TableNote{Each row averages 100 replicates. The quotient singular values diagnose local branch resolvability; Restart SD($\lambda$) denotes the standard deviation across independent starts. RMSE is root mean squared error, and lower values are better for the reported predictive errors.}
\end{table}

Across all 700 replicates, the Spearman correlation between fitted $\sigma_{\min}$ and restart SD($\lambda$) is $0.5308$. A 2,000-sample percentile bootstrap gives interval $[0.4664,0.5930]$. The observed positive direction differs from the prespecified negative ordering. This empirical relationship concerns the finite-sample fitted diagnostic and does not alter the exact population collapse result.

\subsection{Simulation 3: Complementary interaction structures}
\label{supp:misspec-results}

Let $u_i,v_j\stackrel{\mathrm{ind}}\sim\Normal_4(0,I_4)$, independently of Gaussian row and column main effects generated as in Simulation~1. Responses are generated from
\[
\logit(P_{ij})=\mu_{0.30}+\alpha_i+\beta_j+S_{ij},
\qquad
S_{ij}=
\begin{cases}
	-0.7\|u_i-v_j\|, & \text{Euclidean generator},\\
	0.9u_i^\top v_j/\sqrt p, & \text{AME generator},
\end{cases}
\]
where $\mu_{0.30}$ is chosen to give mean event probability $0.30$. Each generator uses $n=m\in\{100,200,400\}$ and 100 replicates per side length. Competing methods receive the same training and test dyads within each replicate. The goal is to assess whether HELPI captures nonadditive dependence when the generating interaction is metric or bilinear rather than hierarchical. Table~\ref{tab:misspecification} averages over side length. Parentheses contain Monte Carlo standard errors across 300 datasets for each generator-model pair.

\begin{table}[H]
	\centering
	\scriptsize
	\caption{Complementary-generator results averaged over side length.}
	\label{tab:misspecification}
	\begin{tabular}{llrrrr}
		\toprule
		Generator & Model & Log score & Brier & AUC & Probability RMSE \\
		\midrule
		Euclidean & HELPI & -0.5708 (0.0006) & 0.1933 (0.0002) & 0.6755 (0.0008) & 0.1073 (0.0004) \\
		Euclidean & Main effects & -0.5701 (0.0006) & 0.1929 (0.0002) & 0.6763 (0.0008) & 0.1058 (0.0004) \\
		Euclidean & Euclidean & -0.5708 (0.0006) & 0.1931 (0.0002) & 0.6756 (0.0008) & 0.1066 (0.0004) \\
		Euclidean & Spherical & -0.5775 (0.0006) & 0.1958 (0.0003) & 0.6679 (0.0009) & 0.1184 (0.0005) \\
		Euclidean & AME & -0.5703 (0.0006) & 0.1930 (0.0002) & 0.6760 (0.0008) & 0.1062 (0.0005) \\
		\addlinespace
		AME & HELPI & -0.5842 (0.0009) & 0.1989 (0.0004) & 0.6445 (0.0016) & 0.1600 (0.0011) \\
		AME & Main effects & -0.5887 (0.0005) & 0.2007 (0.0002) & 0.6346 (0.0007) & 0.1665 (0.0004) \\
		AME & Euclidean & -0.5698 (0.0012) & 0.1923 (0.0005) & 0.6814 (0.0022) & 0.1360 (0.0017) \\
		AME & Spherical & -0.5697 (0.0011) & 0.1922 (0.0004) & 0.6826 (0.0018) & 0.1366 (0.0014) \\
		AME & AME & -0.5725 (0.0013) & 0.1934 (0.0006) & 0.6752 (0.0026) & 0.1391 (0.0020) \\
		\bottomrule
	\end{tabular}
	\TableNote{Parentheses contain Monte Carlo standard errors across 300 datasets for each generator-model pair. AUC is area under the receiver operating characteristic curve, and RMSE is root mean squared error. Higher log score and AUC are better; lower Brier score and RMSE are better.}
\end{table}

Under Euclidean generation, HELPI, main effects, Euclidean, and AME differ by less than $0.00075$ in mean log score. Their predictive performance is therefore similar in this design. Under AME generation, HELPI improves on main effects in log and Brier score. Models that directly accommodate a bilinear signal have the best averages. HELPI is therefore competitive under Euclidean generation and captures nonadditive dependence under AME generation. Relative rankings still depend on the generating interaction.

\subsection{Simulation 4: Analyst-choice sensitivity}
\label{supp:sensitivity-results}

This simulation returns to the balanced hierarchy generator with branching factor $b=3$, depth $d=4$, side length $n=m=200$, hierarchy coefficient $\lambda=0.5$, and mean event probability $\rho=0.30$. The nine design cells cross
\[
p\in\{2,4,8\}
\qquad\text{and}\qquad
R_0\in\{3,5,7\},
\]
with 100 independent replicates per cell.
\begin{samepage}
	For every fitted probability matrix, reporting thresholds $B_{\rm clip}\in\{4,6,8\}$ define
	\[
	\widehat{\mathcal H}_{B_{\rm clip}}
	=\projA\!\left[
	\clip\{\logit(\widehat{\bP}),[-B_{\rm clip},B_{\rm clip}]\}
	\right].
	\]
\end{samepage}
Thus, $p$ and $R_0$ are fitting choices, whereas $B_{\rm clip}$ is a reporting threshold applied to each fitted model. The goal is to assess prediction, hierarchy-signal recovery, clipping, and boundary behavior across the prespecified choices. Table~\ref{tab:sensitivity} reports HELPI means, with Monte Carlo standard errors in parentheses.

\begin{table}[H]
	\centering
	\scriptsize
	\caption{HELPI sensitivity to latent dimension and depth bound.}
	\label{tab:sensitivity}
	\begin{tabular}{rrrrrrr}
		\toprule
		$p$ & $R_0$ & Log score & Brier & AUC & Signal RMSE & Boundary fraction \\
		\midrule
		2 & 3 & -0.5836 (0.0005) & 0.1986 (0.0002) & 0.6460 (0.0008) & 0.41411 (0.00005) & 0 \\
		2 & 5 & -0.5841 (0.0005) & 0.1987 (0.0002) & 0.6456 (0.0008) & 0.41477 (0.00005) & 0 \\
		2 & 7 & -0.5845 (0.0006) & 0.1988 (0.0003) & 0.6439 (0.0010) & 0.41548 (0.00006) & 0 \\
		4 & 3 & -0.5829 (0.0006) & 0.1983 (0.0002) & 0.6463 (0.0009) & 0.41371 (0.00005) & 0 \\
		4 & 5 & -0.5831 (0.0005) & 0.1983 (0.0002) & 0.6464 (0.0008) & 0.41427 (0.00004) & 0 \\
		4 & 7 & -0.5830 (0.0006) & 0.1983 (0.0003) & 0.6463 (0.0008) & 0.41472 (0.00005) & 0 \\
		8 & 3 & -0.5840 (0.0005) & 0.1987 (0.0002) & 0.6458 (0.0009) & 0.41366 (0.00003) & 0 \\
		8 & 5 & -0.5840 (0.0006) & 0.1987 (0.0003) & 0.6443 (0.0008) & 0.41395 (0.00002) & 0 \\
		8 & 7 & -0.5841 (0.0006) & 0.1987 (0.0002) & 0.6464 (0.0009) & 0.41425 (0.00002) & 0 \\
		\bottomrule
	\end{tabular}
	\TableNote{Each cell contains 100 replicates. Parentheses contain Monte Carlo standard errors. $p$ is latent dimension, $R_0$ is the depth bound, and Boundary fraction is the proportion of fitted tangent means within 0.05 of that bound.}
\end{table}

The ranges across the nine cells are narrow for log score, Brier score, AUC, and hierarchy-signal RMSE. The maximum cell-mean clipped fraction at $B_{\rm clip}=4$ is $2.3\times10^{-5}$. At levels six and eight, all cell-mean clipped fractions are zero. No tangent mean lies within $0.05$ of its depth bound. Prediction and hierarchy-signal recovery are therefore stable over this grid, and the finite-depth constraint is inactive in these settings.

\subsection{Junyi preprocessing and primary analysis}
\label{supp:junyi-results}

The preprocessing pass streamed 25,925,992 source rows into a disk-backed first-attempt table. It retained 2,289,789 unique student-exercise pairs before the deterministic core and size limits were applied. The resulting table treats unattempted pairs as missing. Within each student, the earliest 80\% of retained attempts form training data and the later attempts form test data. No test pair was excluded because its exercise was absent from training.

\begin{table}[H]
	\centering
	\scriptsize
	\caption{Junyi analysis subsets after preprocessing.}
	\label{tab:junyi-subsets}
	\begin{tabular}{lrrrrrrr}
		\toprule
		Window & Students & Exercises & Observed pairs & Training & Test & Train rate & Test rate \\
		\midrule
		All days & 2,000 & 500 & 317,091 & 252,888 & 64,203 & 0.6979 & 0.6333 \\
		30 days & 484 & 409 & 18,811 & 14,848 & 3,963 & 0.6977 & 0.6528 \\
		90 days & 990 & 489 & 76,034 & 60,439 & 15,595 & 0.6907 & 0.6334 \\
		180 days & 1,351 & 498 & 131,473 & 104,641 & 26,832 & 0.6954 & 0.6382 \\
		\bottomrule
	\end{tabular}
	\TableNote{Rates are correctness proportions conditional on an observed attempt. Within each student, the earliest 80\% of retained attempts form the training set and later attempts form the test set.}
\end{table}

The deterministic anchor hierarchy has three levels: area, topic, and exercise. The primary subset contains eight areas, 38 area-topic pairs, and 477 exercises with prerequisite metadata. The prerequisite graph is not forced into a unique tree because an exercise may have multiple prerequisite relations. Benchmarks include the Rasch and two-parameter logistic item-response models \citep{birnbaum_1968_LatentTraitModels,rasch_1980_ProbabilisticModelsIntelligence}.

Table~\ref{tab:junyi-primary} reports the primary temporal holdout. Higher log score and AUC are better, while lower Brier score and expected calibration error are better. Calibration intercept zero and slope one are the reference values.

\begin{table}[H]
	\centering
	\scriptsize
	\caption{Primary Junyi correctness results conditional on an observed attempt.}
	\label{tab:junyi-primary}
	\begin{tabular}{lrrrrrr}
		\toprule
		Model & Log score & Brier & AUC & Calibration intercept & Calibration slope & ECE \\
		\midrule
		2PL & -0.5090 & 0.1686 & 0.8075 & 0.1463 & 0.8744 & 0.0220 \\
		HELPI soft & -0.5125 & 0.1701 & 0.8038 & 0.1250 & 0.9264 & 0.0223 \\
		HELPI unanchored & -0.5126 & 0.1702 & 0.8032 & 0.0939 & 0.9608 & 0.0175 \\
		Rasch & -0.5127 & 0.1702 & 0.8037 & 0.1251 & 0.9248 & 0.0224 \\
		HELPI fixed & -0.5130 & 0.1703 & 0.8033 & 0.1236 & 0.9303 & 0.0220 \\
		AME & -0.5165 & 0.1720 & 0.7987 & -0.0064 & 1.0368 & 0.0090 \\
		Euclidean & -0.5173 & 0.1723 & 0.8028 & 0.2604 & 0.8825 & 0.0418 \\
		Spherical & -0.5199 & 0.1731 & 0.7994 & 0.2033 & 0.8558 & 0.0343 \\
		\bottomrule
	\end{tabular}
	\TableNote{The primary split contains 2,000 students, 500 exercises, 252,888 training pairs, and 64,203 test pairs. AUC is area under the receiver operating characteristic curve, and ECE is expected calibration error. Higher log score and AUC are better; lower Brier score and ECE are better. Calibration intercept zero and slope one are reference values.}
\end{table}

Softly anchored HELPI has the best results among the HELPI variants. Across the three variants, log score differs from Rasch by at most $0.0003$, Brier score by at most $0.0001$, and AUC by at most $0.0005$. The 2PL fit has the best primary log score, Brier score, and AUC, while AME has the smallest expected calibration error. HELPI therefore provides predictions close to established item-response models together with a hierarchy-based summary.

Figure~\ref{fig:junyi-model-comparison} displays the primary held-out metrics for every fitted model.

\begin{figure}[H]
	\centering
	\includegraphics[width=\textwidth]{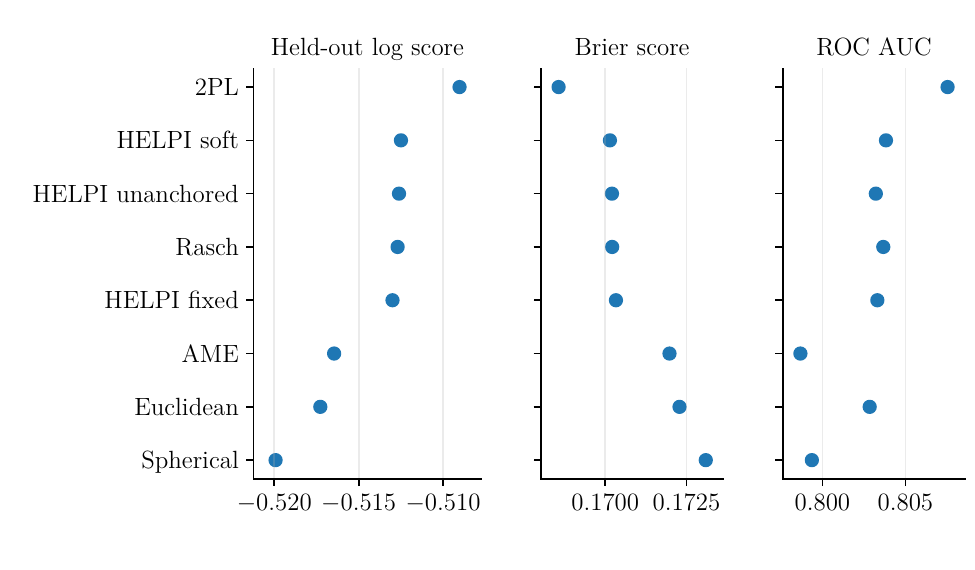}
	\caption{Primary Junyi temporal-holdout comparison. Three aligned dot plots compare eight models by held-out log score, Brier score, and area under the receiver operating characteristic curve under the common split. Each row represents one fitted model. Higher log score and AUC are better, while lower Brier score is better.}
	\label{fig:junyi-model-comparison}
\end{figure}

Figure~\ref{fig:junyi-calibration} shows decile calibration for the primary split. The curves agree with the expected calibration errors in Table~\ref{tab:junyi-primary}. AME follows the diagonal most closely overall, while several models show larger departures in the highest-probability bins.

\begin{figure}[H]
	\centering
	\includegraphics[width=0.78\textwidth]{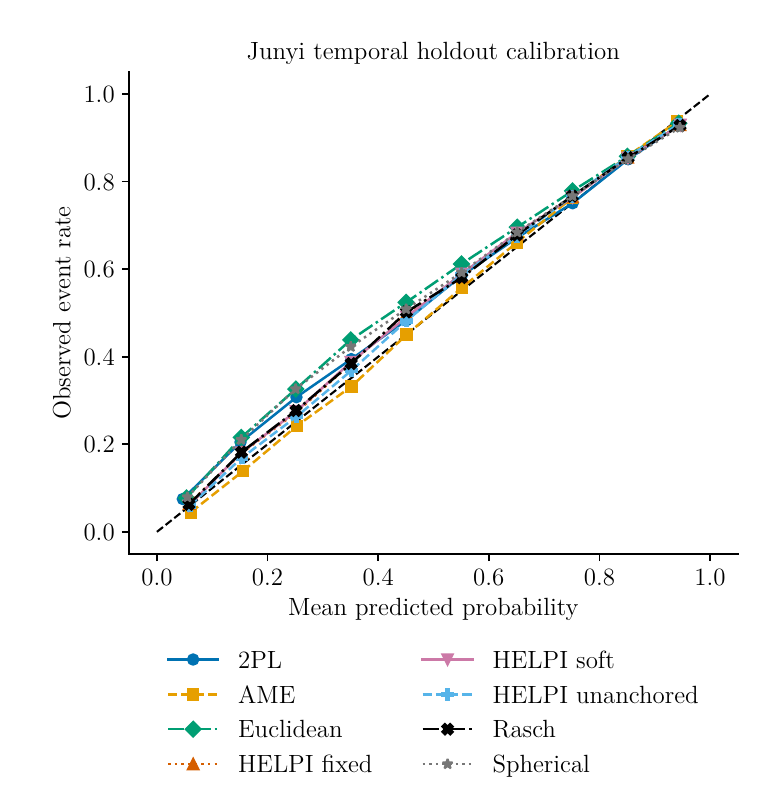}
	\caption{Primary Junyi temporal-holdout calibration. Eight curves plot observed event rates against mean predicted probabilities across ten bins. Points are observed event rates within bins defined by predicted probability. The dashed diagonal denotes exact calibration, and several curves depart from it in the highest-probability bins.}
	\label{fig:junyi-calibration}
\end{figure}

Table~\ref{tab:junyi-diagnostics} reports the primary HELPI identification and boundary diagnostics. Each local diagnostic matrix has 114 columns and rank 110. The full-rank condition in the local branch-identification theorem therefore fails for every anchor specification. We consequently do not interpret the hierarchy coefficient separately.

\begin{table}[H]
	\centering
	\scriptsize
	\caption{Primary HELPI diagnostics and hierarchy-signal sensitivity.}
	\label{tab:junyi-diagnostics}
	\begin{tabular}{lrrrrrrr}
		\toprule
		Model & Rank & Quotient $\sigma_{\min}$ & Restart SD($\lambda$) & Posterior SD($\lambda$) & Boundary & $\|\mathcal H_6\|$ & Clip fraction \\
		\midrule
		Fixed & 110/114 & $1.02\times10^{-10}$ & 0.09380 & 0.01055 & 0 & 0.02429 & 0.000705 \\
		Soft & 110/114 & $4.01\times10^{-5}$ & 0.01676 & 0.00547 & 0 & 0.02691 & 0.000889 \\
		Unanchored & 110/114 & 0.00238 & 0.00383 & 0.02764 & 0 & 0.02340 & 0.000460 \\
		\bottomrule
	\end{tabular}
	\TableNote{Rank reports the local diagnostic rank out of 114 columns. SD denotes standard deviation across restarts or under the fitted posterior. $\mathcal H_6$ uses predictor clipping at six; Clip fraction is the fraction of fitted predictors beyond that threshold.}
\end{table}

At clipping level four, the three fitted hierarchy norms range from $0.09155$ to $0.10856$, with clipped fractions from $0.02069$ to $0.02631$. At level eight, the norms range from $0.01790$ to $0.01970$, and every clipped fraction is below $2.4\times10^{-5}$. This dependence on clipping confirms that the reported functional must include its clip level.

\subsection{Junyi curricular-window sensitivity}

Table~\ref{tab:junyi-windows} reports all model results for the 30-day, 90-day, and 180-day subsets. Rasch has the best log score at 30 and 90 days, while 2PL is best at 180 days. The log-score difference between the best HELPI variant and the best model is $0.0067$, $0.0013$, and $0.0016$ across these windows.

\begin{longtable}{rlrrrr}
	\caption{Junyi curricular-window sensitivity.}\label{tab:junyi-windows}\\
	\toprule
	Days & Model & Log score & Brier & AUC & ECE \\
	\midrule
	\endfirsthead
	\toprule
	Days & Model & Log score & Brier & AUC & ECE \\
	\midrule
	\endhead
	30 & Rasch & -0.5110 & 0.1679 & 0.8027 & 0.0360 \\
	30 & 2PL & -0.5133 & 0.1690 & 0.7991 & 0.0316 \\
	30 & Spherical & -0.5172 & 0.1704 & 0.7969 & 0.0293 \\
	30 & HELPI fixed & -0.5177 & 0.1711 & 0.7953 & 0.0337 \\
	30 & Euclidean & -0.5177 & 0.1706 & 0.7957 & 0.0316 \\
	30 & AME & -0.5183 & 0.1714 & 0.7952 & 0.0363 \\
	30 & HELPI unanchored & -0.5197 & 0.1717 & 0.7937 & 0.0401 \\
	30 & HELPI soft & -0.5201 & 0.1717 & 0.7943 & 0.0394 \\
	\addlinespace
	90 & Rasch & -0.5058 & 0.1671 & 0.8103 & 0.0128 \\
	90 & 2PL & -0.5060 & 0.1671 & 0.8100 & 0.0112 \\
	90 & HELPI fixed & -0.5071 & 0.1677 & 0.8090 & 0.0102 \\
	90 & AME & -0.5085 & 0.1684 & 0.8081 & 0.0162 \\
	90 & HELPI soft & -0.5091 & 0.1684 & 0.8075 & 0.0108 \\
	90 & HELPI unanchored & -0.5093 & 0.1686 & 0.8068 & 0.0102 \\
	90 & Euclidean & -0.5098 & 0.1688 & 0.8069 & 0.0114 \\
	90 & Spherical & -0.5169 & 0.1716 & 0.8010 & 0.0216 \\
	\addlinespace
	180 & 2PL & -0.5059 & 0.1673 & 0.8081 & 0.0199 \\
	180 & HELPI soft & -0.5075 & 0.1683 & 0.8058 & 0.0153 \\
	180 & Rasch & -0.5083 & 0.1683 & 0.8062 & 0.0179 \\
	180 & HELPI fixed & -0.5088 & 0.1687 & 0.8048 & 0.0142 \\
	180 & HELPI unanchored & -0.5089 & 0.1687 & 0.8047 & 0.0137 \\
	180 & Euclidean & -0.5116 & 0.1698 & 0.8036 & 0.0271 \\
	180 & AME & -0.5118 & 0.1699 & 0.8016 & 0.0103 \\
	180 & Spherical & -0.5150 & 0.1711 & 0.8004 & 0.0244 \\
	\bottomrule
\end{longtable}
\TableNote{Each entry is computed on the temporal holdout for the stated curricular window. AUC is area under the receiver operating characteristic curve, and ECE is expected calibration error. Higher log score and AUC are better; lower Brier score and ECE are better.}

The Junyi observation-process check pairs each observed dyad with two reproducibly sampled unobserved dyads and fits AME.
\label{supp:exposure-results}
The holdout contains 190,274 case-control observations with event fraction $0.3335$. The model has AUC $0.9248$, Brier score $0.1009$, log score $-0.3223$, and expected calibration error $0.0052$. Its calibration intercept is $-0.0255$, and its slope is $0.9598$.

The AUC shows that observed and sampled unobserved pairs differ systematically. These results do not estimate population attempt probabilities because the control prevalence was fixed by design \citep{prentice_1979_LogisticDiseaseIncidence}. This observation-process check therefore supports retaining the conditional-on-attempt interpretation for the Junyi correctness analysis.

\subsection{Additional real-data screening and analysis}
\label{supp:additional-realdata}

\subsubsection{Candidate screening}

Candidate selection emphasized public binary bipartite arrays with reproducible acquisition, meaningful row and column roles, and external hierarchical metadata.
We screened the Retailrocket event log, Amazon Reviews 2023 Digital Music, and Voteview congressional roll calls \citep{zykov_2022_RetailrocketRecommenderSystem,hou_2024_BridgingLanguageItems,lewis_2026_VoteviewCongressionalRollCall}.
The goal was to identify applications that add substantive interpretation and test HELPI beyond the educational setting.

Voteview provides policy and bill metadata for a repeated binary voting array.
Retailrocket links product views and transactions to a category tree.
Both datasets therefore support prespecified column anchors and held-out predictive analysis.

The Amazon screen contained 130,434 reviews from 100,952 users on 70,511 products.
Only seven of 70,537 Digital Music metadata records supplied a nonempty category path.
The available release therefore lacked sufficient hierarchy coverage for the planned item analysis.

Voteview was selected for the main manuscript because it supports repeated masks, substantive policy interpretation, and direct visualization of party and policy structure.
Retailrocket was retained as a supplementary temporal and rare-event stress test.
The Amazon decision concerns metadata suitability, not fitted HELPI performance.

\subsubsection{Voteview complete results}

The Voteview analysis uses 49,591 recorded yea or nay votes from 432 members on 118 policy-annotated roll calls in the 117th House.
Five seeded masks assign 80\% of observed dyads to training while retaining every active row and column.
Column anchors follow policy area, bill type, and roll-call identity.
The goal is to compare predictive structure and assess how anchoring balances official policy labels with outcome-defined coalitions.

Table~\ref{tab:voteview-full} includes all eight models, including fixed-anchor HELPI.
Figure~\ref{fig:voteview-model-comparison} shows every mask-specific AUC and Brier score.

\begin{table}[H]
	\centering
	\small
	\caption{Complete Voteview held-out comparison across five independent masks. Entries are means with standard deviations in parentheses.}
	\label{tab:voteview-full}
	\begin{tabular}{lrrr}
\toprule
Model & AUC & Brier score & Log score \\
\midrule
HELPI, soft anchors & $0.9902\;(0.0003)$ & $0.0308\;(0.0004)$ & $-0.1145\;(0.0008)$ \\
HELPI, unanchored & $0.9920\;(0.0006)$ & $0.0289\;(0.0008)$ & $-0.1060\;(0.0019)$ \\
HELPI, fixed anchors & $0.9314\;(0.0016)$ & $0.0942\;(0.0010)$ & $-0.3030\;(0.0033)$ \\
Main effects & $0.9093\;(0.0035)$ & $0.1108\;(0.0013)$ & $-0.3478\;(0.0030)$ \\
2PL & $0.9708\;(0.0013)$ & $0.0545\;(0.0008)$ & $-0.1661\;(0.0028)$ \\
Euclidean & $0.9925\;(0.0006)$ & $0.0278\;(0.0011)$ & $-0.0995\;(0.0024)$ \\
Spherical & $0.9915\;(0.0006)$ & $0.0289\;(0.0008)$ & $-0.1013\;(0.0023)$ \\
AME & $0.9930\;(0.0002)$ & $0.0276\;(0.0005)$ & $-0.0946\;(0.0011)$ \\
\bottomrule
\end{tabular}

	\TableNote{Entries are means over five independent 80\% training masks, with standard deviations in parentheses. AUC is area under the receiver operating characteristic curve. Higher log score and AUC are better; lower Brier score is better.}
\end{table}

\begin{figure}[H]
	\centering
	\includegraphics[width=\textwidth]{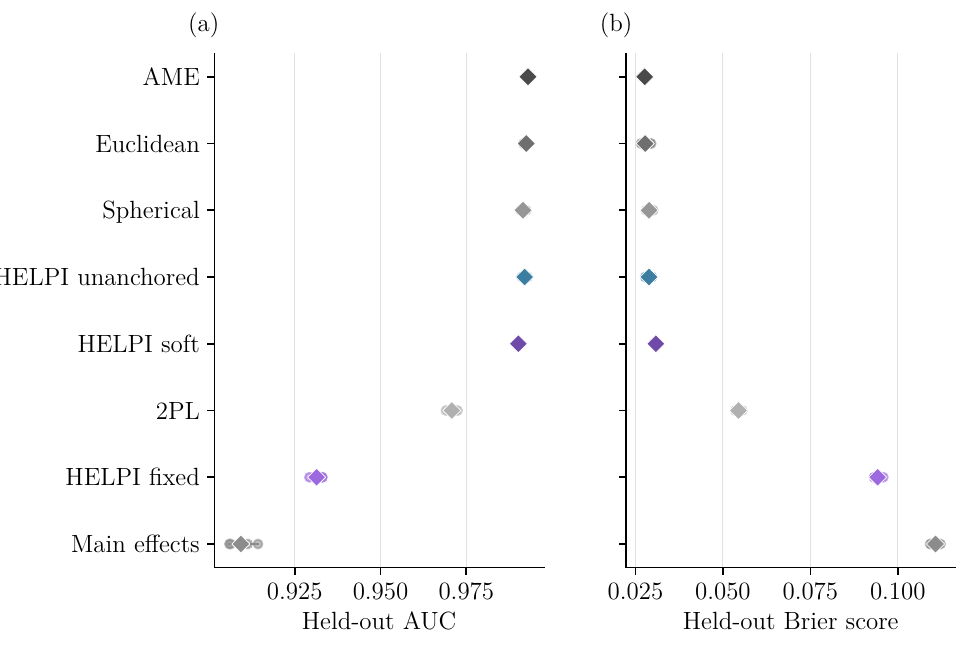}
	\caption{Voteview held-out performance across five masks. Two horizontal dot plots compare eight models by area under the receiver operating characteristic curve in panel (a) and Brier score in panel (b). Points show mask results, and bars show mask means. Higher AUC and lower Brier score are better.}
	\label{fig:voteview-model-comparison}
\end{figure}

Soft-anchor HELPI averages $0.9902$ in AUC and $0.0308$ in Brier score.
Relative to main effects, its mean differences are $0.0809$ in AUC and $-0.0800$ in Brier score.
Relative to 2PL, the corresponding differences are $0.0194$ and $-0.0236$.
Unanchored HELPI reaches mean AUC $0.9920$ and Brier score $0.0289$.
Fixed anchors also improve on main effects, with mean AUC $0.9314$ and Brier score $0.0942$.

Each HELPI specification captures predictive voting structure beyond additive member and roll-call propensities.
Soft anchors preserve policy organization while accommodating coalitions that are not fully encoded by the official labels.
Relative to the AME and Euclidean fits, unanchored HELPI differs by at most $0.0010$ in mean AUC and $0.0013$ in mean Brier score.
The fixed-anchor improvement over main effects shows that the external hierarchy contains predictive information, while the flexible variants capture additional organization.

\subsubsection{Retailrocket temporal conversion analysis}

The Retailrocket analysis begins with each retained visitor-product first view and records whether a later transaction occurs for that pair.
A within-visitor temporal split assigns earlier views to training and later views to testing.
The retained array contains 20,834 pairs from 1,443 visitors and 800 products.
It includes 16,078 training pairs and 4,756 test pairs.
Column anchors follow root category, local parent, leaf category, and product.
The goal is to evaluate rare-event conversion prediction under a realistic temporal shift.

Training purchase prevalence is $0.0624$, compared with $0.0467$ in testing.
Table~\ref{tab:retailrocket-summary} reports the temporal holdout, and Figure~\ref{fig:retailrocket-performance} adds calibration curves.

\begin{table}[H]
	\centering
	\small
	\caption{Retailrocket temporal-holdout performance for purchase after an observed product view.}
	\label{tab:retailrocket-summary}
	\begin{tabular}{lrrr}
\toprule
Model & AUC & Brier score & Log score \\
\midrule
HELPI, soft anchors & $0.7715$ & $0.0462$ & $-0.1823$ \\
HELPI, unanchored & $0.7693$ & $0.0466$ & $-0.1835$ \\
HELPI, fixed anchors & $0.7657$ & $0.0469$ & $-0.1848$ \\
Main effects & $0.7939$ & $0.0470$ & $-0.1777$ \\
2PL & $0.7648$ & $0.0528$ & $-0.2019$ \\
Euclidean & $0.7673$ & $0.0442$ & $-0.1719$ \\
Spherical & $0.7719$ & $0.0446$ & $-0.1725$ \\
AME & $0.7676$ & $0.0480$ & $-0.1850$ \\
\bottomrule
\end{tabular}

	\TableNote{Entries use one within-visitor temporal split with 16,078 training pairs and 4,756 test pairs. AUC is area under the receiver operating characteristic curve. Higher log score and AUC are better; lower Brier score is better.}
\end{table}

\begin{figure}[H]
	\centering
	\includegraphics[width=\textwidth]{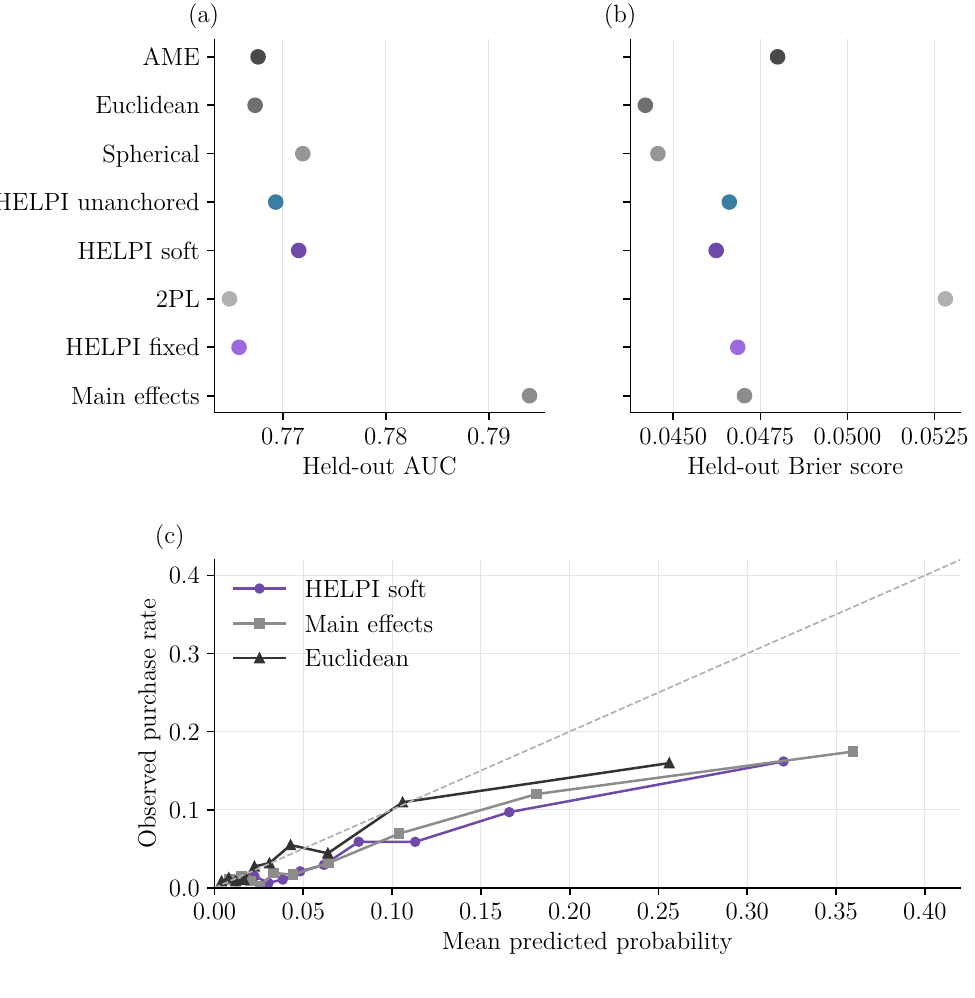}
	\caption{Retailrocket temporal-holdout results. Panels (a) and (b) are horizontal dot plots comparing eight models by AUC and Brier score. Panel (c) plots calibration curves for soft-anchor HELPI, main effects, and Euclidean interaction against a diagonal reference under a temporal split with lower test prevalence.}
	\label{fig:retailrocket-performance}
\end{figure}

Soft-anchor HELPI has the highest AUC among HELPI variants at $0.7715$.
Its Brier score is $0.0462$, compared with $0.0470$ for main effects and $0.0528$ for 2PL.
Euclidean interaction gives the lowest Brier and best log score, while main effects gives the highest AUC.
The calibration curves also show that predictions fitted under the higher training prevalence tend to exceed observed test rates.

The Retailrocket results show that HELPI captures predictive structure in a sparse conversion setting with a category tree and temporal prevalence shift.
Soft anchoring has the best AUC among HELPI variants and improves the Brier score over the additive benchmark.
Because the analysis uses one temporal split, it does not quantify variation across splits and is treated as a complementary stress test.
\subsection{MCMC convergence audit}
\label{supp:mcmc-results}

The computational audit covers one $35\times35$ balanced simulation and two connected real-data subsets.
The Junyi subset contains 50 students and 35 exercises.
The Voteview subset contains 50 members and 35 roll calls, with 1,589 training votes and 153 held-out votes.
Moderate subsets keep the per-iteration calculation auditable while preserving both row and column heterogeneity.

The simulation and Junyi targets use four independent chains each. Voteview uses eight independent chains, one on each workstation core. Every chain has 8,000 warmup iterations and 32,000 retained draws, with proposal adaptation ending after warmup. Earlier short runs for the simulation and Junyi targets are archived, and their extensions were determined only from convergence diagnostics. Voteview adopted the resulting schedule before its chains were run.

Figure~\ref{fig:mcmc-diagnostics} and Table~\ref{tab:mcmc-diagnostics} report rank-normalized split $\widehat R$ and effective sample sizes (ESS) \citep{vehtari_2021_RankNormalizationFoldingLocalization}. For each target, the six monitored summaries are invariant to root-fixing rotations. They include selected probabilities, log score, hierarchy norm, hierarchy contrast, and $\lambda$.

\begin{figure}[H]
	\centering
	\includegraphics[width=\textwidth]{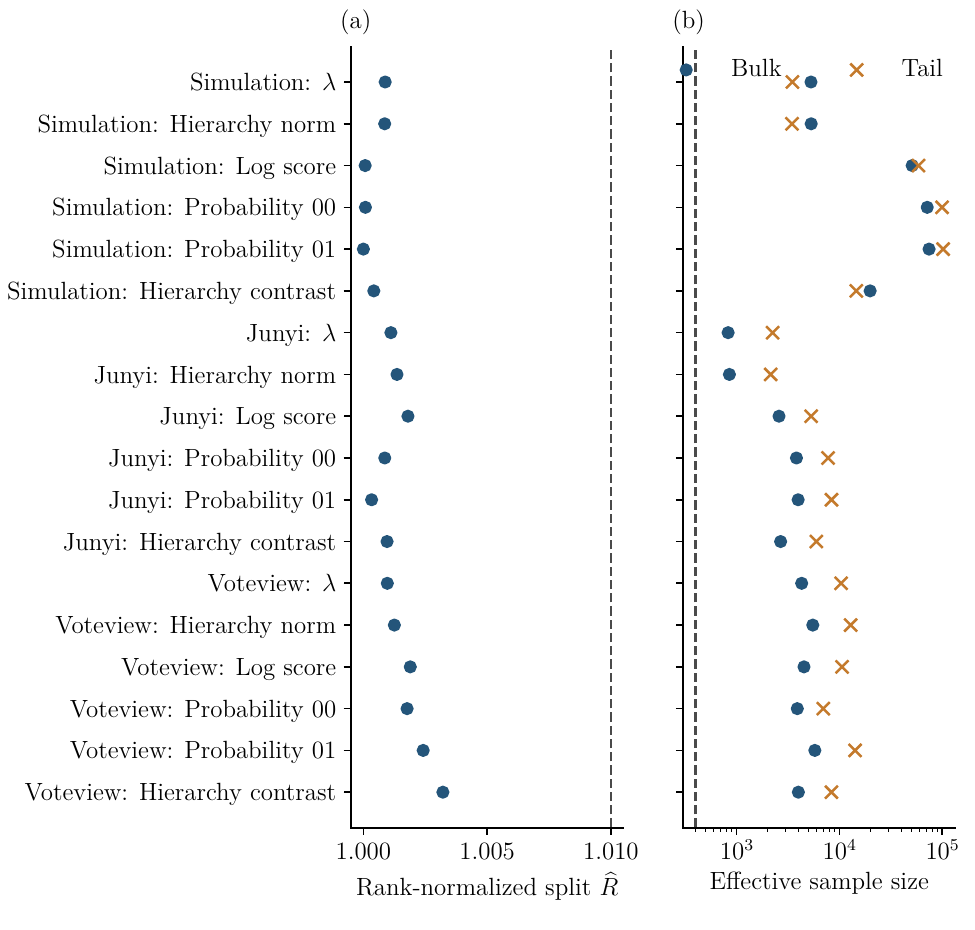}
	\caption{MCMC convergence diagnostics for the simulation, Junyi, and Voteview targets. Panel (a) plots rank-normalized split $\widehat R$ values, while panel (b) plots bulk and tail effective sample sizes on a logarithmic scale. Dashed vertical lines mark the prespecified thresholds of $1.01$ and 400, respectively.}
	\label{fig:mcmc-diagnostics}
\end{figure}

\clearpage
\begingroup
\small
\setlength{\tabcolsep}{4pt}
\begin{longtable}{@{}l@{\hspace{1em}}lrrrrr@{}}
\caption{Rank-normalized MCMC diagnostics.}\label{tab:mcmc-diagnostics}\\
\toprule
Target & Summary & Mean & SD & $\widehat R$ & Bulk ESS & Tail ESS \\
\midrule
\endfirsthead
\toprule
Target & Summary & Mean & SD & $\widehat R$ & Bulk ESS & Tail ESS \\
\midrule
\endhead
Simulation & $\lambda$ & 0.3139 & 0.3137 & 1.0009 & 5,300.1 & 3,498.7 \\
Simulation & Hierarchy norm & 0.0827 & 0.0829 & 1.0009 & 5,324.7 & 3,464.2 \\
Simulation & Log score & -0.6689 & 0.0239 & 1.0001 & 51,095.5 & 58,849.0 \\
Simulation & Probability 00 & 0.1874 & 0.0919 & 1.0001 & 71,570.0 & 99,774.3 \\
Simulation & Probability 01 & 0.2703 & 0.1194 & 1.0000 & 74,532.1 & 102,154.2 \\
Simulation & Hierarchy contrast & -0.0032 & 0.1654 & 1.0004 & 19,925.8 & 14,626.0 \\
\addlinespace
Junyi & $\lambda$ & 6.9754 & 0.9494 & 1.0011 & 828.8 & 2,245.5 \\
Junyi & Hierarchy norm & 1.8674 & 0.2569 & 1.0014 & 853.1 & 2,150.1 \\
Junyi & Log score & -0.1640 & 0.2211 & 1.0018 & 2,590.5 & 5,317.9 \\
Junyi & Probability 00 & 0.9835 & 0.0320 & 1.0009 & 3,828.9 & 7,767.6 \\
Junyi & Probability 01 & 0.9925 & 0.0167 & 1.0003 & 3,976.1 & 8,390.3 \\
Junyi & Hierarchy contrast & -0.1903 & 2.7987 & 1.0010 & 2,686.1 & 5,974.8 \\
\addlinespace
Voteview & $\lambda$ & 12.4986 & 0.8462 & 1.0010 & 4,304.4 & 10,399.1 \\
Voteview & Hierarchy norm & 4.2355 & 0.2898 & 1.0013 & 5,518.3 & 12,888.6 \\
Voteview & Log score & -0.1699 & 0.0341 & 1.0019 & 4,533.6 & 10,633.9 \\
Voteview & Probability 00 & 0.9859 & 0.0387 & 1.0018 & 3,900.9 & 6,986.0 \\
Voteview & Probability 01 & 0.9759 & 0.0680 & 1.0024 & 5,778.3 & 14,229.1 \\
Voteview & Hierarchy contrast & -5.9706 & 4.5704 & 1.0032 & 3,995.8 & 8,379.2 \\
\bottomrule
\end{longtable}

\TableNote{Mean and SD summarize retained posterior draws. ESS is effective sample size. The audit uses four chains each for the simulation and Junyi targets and eight chains for Voteview; the prespecified thresholds are $\widehat R\le1.01$ and bulk and tail ESS at least 400.}
\endgroup

All 18 diagnostics pass the prespecified thresholds.
The combined maximum $\widehat R$ is $1.0032$, the minimum bulk ESS is $828.8$, and the minimum tail ESS is $2,150.1$.
Within Voteview, the maximum $\widehat R$ is $1.0032$.
Its minimum bulk and tail effective sample sizes are $3,900.9$ and $6,986.0$, respectively.

The diagnostics support convergence for the moderate computational targets. They do not replace the full-data predictive comparisons or establish variational interval coverage. A coverage study would require common invariant functionals and intervals from both procedures.

\bibliography{references-supplement}

\end{document}